\documentclass[journal,10pt]{IEEEtran}

\usepackage{amsmath,amsfonts}
\usepackage{tikz}
\usepackage{amsthm, thmtools}
\usepackage{mathtools}
\usepackage[caption=false,font=normalsize,labelfont=sf,textfont=sf]{subfig}
\usepackage{graphicx}
\usepackage{cleveref}
\usepackage{booktabs}
\usepackage{paralist}
\usepackage{bbm}
\usepackage{bm}
\usepackage{comment}
\crefname{observation}{observation}{observations}
\crefname{algorithm}{algorithm}{algorithms}
\crefname{align}{equation}{equations}
\crefname{eqnarray}{equation}{equations}
\crefname{section}{Sec.}{Sections}
\crefname{figure}{Fig.}{Figures}
\crefname{theorem}{Theorem}{Theorems}
\crefname{lemma}{Lemma}{Lemmas}
\crefname{proposition}{Proposition}{Propositions}
\crefname{assumption}{assumption}{assumptions}
\crefname{condition}{condition}{conditions}
\crefname{conjecture}{Conjecture}{Conjectures}

\newtheorem{theorem}{Theorem}
\newtheorem{lemma}{Lemma}

\newtheorem{proposition}{Proposition}

\newtheorem{conjecture}{Conjecture}

\theoremstyle{definition}
\newcounter{statement}
\newtheorem{definition}{Definition}

\newtheorem{assumption}[statement]{Assumption}

\newtheorem{discussion}{Discussion} 

\theoremstyle{remark}
\newtheorem*{remark}{Remark}

\newcommand{\remove}[1]{}

\DeclareMathOperator*{\argmax}{\arg\!\max}

\DeclareMathOperator*{\esssup}{esssup}

\newcommand{\expect}[1]{\mathbb{E}\left[{#1}\right]}
\newcommand{\expectt}[2]{\mathbb{E}_{#1}\left[{#2}\right]}

\newcommand{\barexpectt}[2]{\overline{\mathbb{E}}_{#1}\left[{#2}\right]}
\newcommand{\indicator}[1]{\mathbf{1}_{\left\{{#1}\right\}}}
\newcommand{\stochleq}{\leq_{st}}
\newcommand{\lrleq}{\underset{LR}{\leq}}
\newcommand{\N}{\mathbb{N}}

\newcommand{\smallO}[1]{o({#1})}

\newcommand{\kldist}[2]{I\left({#1}, {#2}\right)}

\newcommand{\notes}[1]{}

\newcommand{\first}[1]{$1^{\mathrm{st}}$}
\newcommand{\second}[1]{$2^{\mathrm{nd}}$}

\newcommand\independent{\protect\mathpalette{\protect\independenT}{\perp}}
\def\independenT#1#2{\mathrel{\rlap{$#1#2$}\mkern2mu{#1#2}}}
\newcommand{\squishlisttwo}{
\begin{list}{$\blacktriangleright$}
{ \setlength{\itemsep}{0.5pt}
\setlength{\parsep}{0pt}
\setlength{\topsep}{0pt}
\setlength{\partopsep}{0.5pt}
\setlength{\leftmargin}{1em}
\setlength{\labelwidth}{1em}
\setlength{\labelsep}{0.5em} } }

\newcommand{\squishend}{
\end{list} }
\allowdisplaybreaks[1]

\def\bm#1{\mbox{\boldmath $#1$}}

\def\bSigma{{
      \ooalign{
      \smash{\hskip.4pt\raise.4pt\hbox{$\Sigma$}}\vphantom{}\crcr
      \smash{\hskip.7pt\raise.6pt\hbox{$\Sigma$}}\vphantom{}\crcr
      \smash{\hbox{$\Sigma$}}\vphantom{$\Sigma$}}
      \vphantom{\hbox{$\Sigma$}}
      }}

\newcommand{\CUSUM}{CUSUM}

\begin{document}

\title{Quickest Change Point Detection with Measurements over a Lossy Link}

\author{Krishna Chaythanya KV, Saqib Abbas Baba, Anurag Kumar, Arpan Chattopadhyay, Rajesh Sundaresan
  \thanks{This work was done while Krishna Chaythanya KV was at the Indian Institute of Science (IISc), Bengaluru and Saqib Abbas Baba was at the Indian Institute of Technology, New Delhi. Both of them are co-first authors.}
  \thanks{Anurag Kumar and Rajesh Sundaresan are with the Department of Electrical Communication Engineering at the Indian Institute of Science, Bengaluru, and Arpan Chattopadhyay is with the Electrical Engineering Department and the Bharti School of Telecommunication Technology and Management, Indian Institute of Technology Delhi, India}
  \thanks{Krishna Chaythanya KV was supported by a fellowship grant from the Centre for Networked Intelligence (a Cisco CSR initiative) of the Indian Institute of Science, Bengaluru.}
  \thanks{Anurag Kumar acknowledges support from an Indian National Science Academy Distinguished Professorship.}
  \thanks{Arpan Chattopadhyay acknowledges support from the Science and Engineering Research Board (SERB), India, through grant no. CRG/2022/003707.}
  \thanks{Rajesh Sundaresan acknowledges support from the SERB through grant no. CRG/2019/002975, and from the Centre for Networked Intelligence (a CISCO CSR initiative) of the Indian Institute of Science, Bangalore}
  \thanks{A preliminary version of this work was published as a conference paper~\cite{qcd_lossy_link_conference}. This work extends the conference version significantly as indicated towards the end of Section~\ref{sec:introduction}.}
}

\maketitle

\begin{abstract}
Motivated by Industry 4.0 applications, we consider quickest change detection (QCD) of an abrupt change in a process when its measurements are transmitted by a sensor over a lossy wireless link to a decision maker (DM). The sensor node samples measurements using a Bernoulli sampling process, and places the measurement samples in the transmit queue of its transmitter. The transmitter uses a retransmit-until-success transmission strategy to deliver packets to the DM over the lossy link, in which the packet losses are modeled as a Bernoulli process, with different loss probabilities before and after the change. We pose the QCD problem in the non-Bayesian setting under Lorden’s framework, and propose a CUSUM algorithm. By defining a suitable Markov process, involving the DM measurements and the queue length process, we show that the problem reduces to QCD in a Markov process. Characterizing the information measure per measurement sample at the DM, we establish the asymptotic optimality of our algorithm when the false alarm rate tends to zero. Further, when the DM receives incomplete data due to channel loss, we present asymptotically optimal QCD algorithms by suitably modifying the CUSUM algorithm. We then explore the last-come-first-served (LCFS) queuing discipline at the sensor transmit queue to lower detection delay in the non-asymptotic case. Next, we consider the case of multiple sensors, each with its own wireless transmitter queue, and show that our analysis extends to the case of multiple homogeneous sensors. When the sensors are heterogeneous, i.e., their observations are not identically distributed, we present a sensor scheduling algorithm that minimizes detection delay by balancing the trade-off between the age of the observations and their information content. Numerical analysis demonstrate trade-offs that can be used to optimize system design parameters in the non-asymptotic regime.
\end{abstract}


\section{Introduction}\label{sec:introduction}
Online monitoring of industrial systems is a cornerstone of predictive maintenance and the broader vision of Industry~4.0~\cite{raiReviewSignalProcessing2016}. Timely detection of incipient faults, such as bearing degradation, enables corrective action before catastrophic failure and costly downtime. In many such deployments, sensors must communicate wirelessly with a central decision maker (DM), since wired connections are often impractical for moving machinery. Consequently, sensor measurements may experience packet losses or retransmissions over unreliable wireless links. Moreover, the same physical degradation that affects the measured process (e.g., increased vibration) may also deteriorate the wireless channel quality, coupling sensing and communication dynamics.

Against this backdrop, in this paper we study the classical problem of quick detection of a change (QCD) in a stochastic process with the novel feature that the sensor measurements can experience random loss, and, therefore, delay due to retransmissions. In addition, when the change in the machine related process is concomitant with a degradation of the wireless channel, the packet loss process also provides information about the change, in addition to the contents of the delivered measurement packets. The objective is to detect an abrupt distributional change in the underlying process as quickly as possible, subject to a constraint on false alarms. 

We first consider the case where a single sensor monitors the process and transmits its measurements over a lossy wireless channel. By augmenting the observation space with the queue length process, we formulate the detection problem as one over a Markov process and prove that the proposed \CUSUM{}-based detector remains asymptotically optimal. We also explore the Last-Come-First-Served (LCFS) discipline to
lower detection delay. We then extend this framework to a multi-sensor setting, where multiple sensors independently monitor a process, and transmit their observations over a shared wireless channel.  This introduces the problem of sensor scheduling, where the order of service affects both the freshness and informativeness of received data. In non-homogeneous networks, sensors may have different sampling rates and signal-to-noise levels, leading to heterogeneous information contributions. We propose heuristic scheduling policies, including a discounted-information rule and a look-back window policy, that jointly balance data recency and information quality. 

Finally, we numerically evaluate the proposed algorithms in both single- and multi-sensor setups, highlighting the effect of queueing delays, channel losses, and sampling rates on the average detection delay. The analysis demonstrates that an optimal sampling rate exists that minimizes detection delay in the non-asymptotic regime. Overall, our formulation of QCD for Markov process, combined with the inclusion of wireless queuing and scheduling effects, offers a unified framework bridging detection theory and networked sensing.

\subsection{Related Literature}
\label{networked-cusum-literature-survey}

The quickest change detection (QCD) problem, in both Bayesian and non-Bayesian formulations, has been extensively studied in the classical literature. In the Bayesian setting, the change time is modeled as a random variable with a known prior, leading to optimal detection schemes such as the Shiryaev \cite{shiryaev1963} and Shiryaev–Roberts tests. In the non-Bayesian minimax framework of Lorden~\cite{lordenProceduresReactingChange1971} and Pollak \cite{pollakOptimalDetectionChange1985}, the objective is to minimize the worst-case average detection delay (ESADD) under a constraint on the average run length to false alarm (ARL2FA). In this setting, the \CUSUM{} test has been shown to be asymptotically optimal~\cite{laiInformationBoundsQuick1998, moustakidesOptimalStoppingTimes1986}. Comprehensive surveys can be found in~\cite{veeravallietal2021, tartakovsky2014sequential, tartakovsky2019sequential}.

Beyond the i.i.d.\ setting, sequential change detection for non-i.i.d.\ processes has been investigated under a variety of models. For example, \cite{yakir-cpd-markov-finite-ss} studies QCD for finite-state Markov processes, while \cite{fuhSPRTCUSUMHidden2003} analyzes \CUSUM{} performance for hidden Markov models. Moustakides et al.~\cite{moustakidesMarkovComments2010} and \cite{moustakides2015} further explore the optimality of \CUSUM{} for Markovian data and propose variants such as the Shewhart test optimized for worst-case detection probability. These works establish foundational results for QCD in dependent-observation models, which our work extends to the queue-augmented observation processes induced by wireless communication.

In recent years, attention has turned toward QCD in decentralized and networked sensing environments, where observations from multiple sensors must be communicated over shared channels to a fusion center. In \cite{karumbuDelayOptimalEvent2012}, event detection over ad hoc wireless sensor networks was analyzed in a Bayesian framework, revealing the trade-off between random network delays and detection delay under random-access contention. Banerjee and Veeravalli~\cite{banerjee2015data} studied data-efficient distributed QCD using on–off observation control and censoring policies to limit communication cost while maintaining asymptotic optimality. Premkumar and Kumar~\cite{premkumar2008optimal} examined optimal sleep–wake scheduling for intrusion detection in sensor networks, focusing on energy-efficient sampling strategies. Active control of sampling for QCD has been discussed in \cite{xu2021optimum,liuAdaptiveSamplingStrategy2015,xuSecondOrderAsymptoticallyOptimal2020}. Similarly, Ren et al.~\cite{ren2016quickest} proposed a threshold-based observation scheduling policy that dynamically selects between multiple observation modes with different costs and information quality. These studies address communication constraints and observation control but generally do not model queueing or retransmission effects at the network layer.

In summary, while prior works have explored observation control, scheduling, and communication-aware detection, they largely assume instantaneous or complete data availability at the decision maker. In practical wireless systems, however, measurement packets can arrive asynchronously or be delayed due to retransmissions, leading to partially observed data streams. The interplay between such queueing and retransmission effects and sequential detection performance has not been analytically characterized. 

\noindent \subsection*{Our Contributions} This paper addresses the gap highlighted in the previous paragraph by formulating the QCD problem for both single- and multi-sensor wireless settings, establishing asymptotic optimality in the single-sensor case, and proposing heuristic yet effective scheduling policies for the multi-sensor case. A preliminary version of this work appeared in our conference paper~\cite{qcd_lossy_link_conference}, which focused on the single-sensor case. The main contributions of this work are as follows:
\begin{enumerate}
    \item We formulate the non-Bayesian QCD problem under Lorden’s criterion in Section~\ref{sec:infin-retr-in-order-reception} for a single sensor transmitting over a lossy wireless link with retransmissions, and establish asymptotic optimality of a Markovian \CUSUM{} detector.
    \item We analyze how transmit-queue service order influences detection delay in Section~\ref{sec:non-fcfs-discipline}. We provide a rigorous discussion of the impact of transmission queue discipline, including FCFS and LCFS, on detection delay, and provide heuristic arguments in favor of LCFS among non-idling policies within a busy period.
    \item We extend the framework to multi-sensor networks sharing a wireless channel. For homogeneous sensors, we derive the corresponding \CUSUM{} formulation in Section~\ref{sec:multi-sens}; for heterogeneous sensors, we highlight the trade-off between information content and data freshness.
    \item We propose two heuristic scheduling rules in Sections~\ref{subsec:dis_info_sch_pol} and~\ref{subsec:look_back_win_pol}, a discounted-information and a look-back window policy, that adapt transmission based on informativeness and recency, demonstrating their effectiveness in reducing non-asymptotic detection delay via simulation results.
\end{enumerate}

\section{Single Sensor Case: System Model and Notation}
\label{sec:system-model}
We consider a discrete-time system where a sensor node samples a random process at a sampling rate $0 < r < 1$ per slot, i.e., a new sample is generated in each discrete time interval with a probability $r$. The sample is encapsulated in a packet, and immediately added to the transmit queue of the transmitter of a wireless link connecting the sensor to the decision maker. The wireless channel is time-slotted, and memoryless, with a known packet loss probability. If its queue is nonempty at the beginning of a slot, the transmitter transmits one packet; if the packet is successfully received at the DM, an acknowledgment is received back in the same slot, else the packet is backlogged for reattempt in the next slot. We assume that the acknowledgment packets do not undergo any packet loss over the feedback channel, and are always received by the transmitter, whenever the DM acknowledges the received packets. The time slots are of unit size (in practice the time taken to transmit one packet and receive its acknowledgment, along with the inter-packet gaps and guard times) and are indexed by $k\in\mathbb{Z} = \left\{ \ldots,-1, 0, 1, 2, \ldots \right\}$, where slot $k$ refers to the time interval $\left[  k-1, k\right)$. We assume that the nodes in the network (the sensor and the DM in the single sensor case, and all the sensors and the DM in the multi-sensor case) are all time synchronized.

We also assume that the sensor has been generating samples from the random process for an extended duration prior to slot $0$ and that the QCD procedure starts after time $0$.  In practice, time $0$ demarcates a phase of known normal behavior and a subsequent regime where it is anticipated that anomalous behavior in the process under observation may occur at a \emph{change point}.

The sensor node samples a measurement $X_j$ at time denoted $t_j$, where $j=1, 2, \ldots$. The measurements are independent, and have a probability distribution
\begin{displaymath}
    X_j \sim
  \begin{cases}
    f_0 & \text{if }t_j \leq \nu,\\
    f_1 & \text{if }t_j > \nu,
  \end{cases}
\end{displaymath}
where $\nu \geq 1$ is an unknown deterministic time, aligned with the trailing slot boundary, and referred to as the \emph{change point} at which the distribution of the observations changes from a known distribution $f_0$ to a known distribution $f_1$. We assume that this change point occurs at the end of the slot $\nu$ and any measurement in slot $\nu$ observes the pre-change distribution. This change in the distribution of the measurements may occur due to the development of a fault in one of the components of the machinery, whose health is being monitored by the sensor node.  In addition, the channel over which the sensor node transmits to the DM also changes after the change point. The channel has a probability of successful transmission $p_0$ for slot $k \leq \nu$, and $p_1$ for slot $k > \nu$. These probabilities are known to the DM.

\begin{assumption}
  \label{assump:networked-qcd-chan-ind-sampling}
  We assume that the channel is conditionally independent of the sampling process given the change point $\nu$.
\end{assumption}
\begin{assumption}
  \label{assump:sampling-rate-sane}
  We assume that the sampling rate of the sensor node is less than both the pre-change and the post-change probability of successful transmission, i.e., $r < \min\left\{ p_0, p_1 \right\}$.
\end{assumption}
The problem is for the DM to detect the change in the distribution of the samples as quickly as possible, when the DM receives data sequentially, and is aware of the packet loss probabilities of the channel before and after the change, while controlling the false alarms to be below a given threshold. The QCD procedure begins at slot $1$. We first define the notation before we state our problem formally.
\begin{itemize}
\item $\mathbb{P}_\nu, \mathbb{E}_\nu$, for $\nu \geq 1$, denote the underlying probability law
  and the expectation, when the change occurs in the slot $\nu$.
\item $\mathbb{P}_0, \mathbb{E}_0$, denote the probability law, and the expectation, when all the random variables are governed by the
post-change probabilities.
\item $\mathbb{P}_\infty, \mathbb{E}_{\infty}$ denote the probability law, and
  the expectation, when the change does not occur ($\nu=\infty$). Under $\mathbb{P}_\infty$, all the random variables are governed by the pre-change probabilities.
\item $S_\nu$ denotes the number of measurements that arrived at the
  transmit queue after the QCD process starts until the change point
  $\nu$, i.e., it counts the number of arrivals in the time
  $\left(0, \nu\right]$.
\item $Q_k$ is the number of samples in the queue at the beginning of
  the time slot $k$ (see \cref{fig:notationfigure}).  Packet arrivals
  into the queue during the time $\left( k-2, k-1\right]$ are
  accounted in $Q_k$. The sensor node attempts a transmission in slot
  $k$ if $Q_k > 0$.
\item Measurement packets generated by the sensor are numbered sequentially. $D_k$ denotes the sampling slot of the last measurement packet
  successfully received at the DM up to and including the end of slot $k$.
\item $A_k$ is the packet arrival process at the transmit queue. When a packet containing a measurement arrives at the transmit queue in slot $k$, then $A_k = 1$, else $A_k = 0$. All packets that arrive into the queue till the trailing edge of slot $k$ are accounted in the queue length computed in slot $k+1$.
  \item $Y_k$ is the channel service process. It represents whether a transmission attempt in that slot was successful, failed, or not made.
\item $Z_k$ is the measurement received at the DM in slot $k$.  If
  $Y_k = 1$ and the index of the corresponding packet departure from
  the transmitter queue is $D_k$ , then $Z_k = X_{D_k}$. To denote
  that there was no transmission on the channel due to an empty
  transmit queue (i.e., $Q_k = 0$), we say that
  $Y_k = \emptyset$. In cases when
  $Y_k \in \left\{ 0, \emptyset \right\}$, no measurement is received
  at the DM. Thus, $\left( Y_k, Z_k\right)$ are defined as
  \begin{displaymath} (Y_k, Z_k) =
    \begin{cases} (1, X_{D_k}),&\text{on successful transmission},\\
\left( 0, *\right),&\text{on unsuccessful transmission},\\ \left(
\emptyset,*\right),&\text{on no transmission}.
    \end{cases}
  \end{displaymath}
We assume that the value of $Y_k$ (either $\emptyset$ or $0$ or $1$) is known to the receiver. For example, absence of energy could inform that there was no transmission ($Y_k=\emptyset$). If energy is detected, which indicates a transmission, parity check could inform whether the transmitted packet was incorrectly received ($Y_k=0$) or correctly received ($Y_k=1$).
\item The probability of successful transmission over the channel is 
\begin{displaymath}
  P\left( Y_k = 1 \mid Q_k > 0 \right) = 
  \begin{cases}
  p_0 & \text{ if }k \leq \nu,\\
       p_1 & \text{ if }k > \nu.
  \end{cases}
\end{displaymath}
We assume that $p_0, p_1$ are known.
\item $J_k$ denotes the sampling index of the measurement received at the DM in the $k$th slot, whenever a successful reception at the DM occurs. We need to define $J_k$ as we will consider non-FCFS transmission from the measurement queue. In the case that no measurement was received at the DM, i.e., when $Z_k = *$, we denote $J_k = *$.
\end{itemize}
\begin{figure}[h]
\centering
\includegraphics[width=0.85\columnwidth,height=3.7cm]{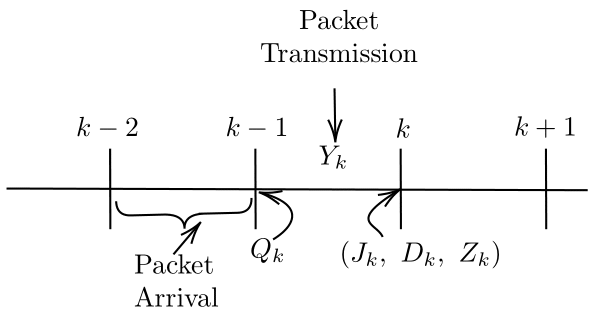}
\caption{Slot $k$ corresponds to time $\in \left[k-1, k\right)$. The
  change point here is taken to arrive at the end of the slot $\nu$.}
    \label{fig:notationfigure}
\end{figure}
We assume that the sensor, transmit queue, and transmitter have been in continuous operation for an extended duration prior to slot~0, so that the system is in steady state when the change detection procedure begins at the start of slot~1. Slot~0 thus represents the last time instant known a priori to be pre-change, implying that the change point satisfies $\nu > 0$. Consequently, the transmit queue may initially contain $Q_1$ undelivered packets sampled before the start of the procedure; if these are later received by the DM, they are discarded since they correspond to pre-change data.

\begin{assumption}\label{assump:networked-qcd-q1-known}
The transmitter queue length at the beginning of the first slot, $Q_1$, is known \emph{a priori} to the DM.
\end{assumption}

This initial queue state, $Q_1$, can be communicated via a control packet that initializes the QCD process. We assume that control packets (initialization, acknowledgments) always succeed (due to they being small, and being transmitted with more robust schemes not costed in our current model).

We begin our analysis with the baseline assumption that all measurements sampled by the sensor are eventually received at the DM in the order of sampling, although with some delay. A \CUSUM{} algorithm is developed for this case and its performance analyzed. In later sections, we investigate cases with out-of-order reception and packet losses at the DM.

\section{Single Sensor Case: In-order reception with network delay}
\label{sec:infin-retr-in-order-reception}

In this section, we investigate the setting where the retransmission protocol only introduces delays in the reception of measurements at the DM and never results in the loss of measurement packets. Further, the measurement packets are received at the DM in-order. We make the following assumptions on the system model described in
\cref{sec:system-model}.
\begin{enumerate}
\item The transmitter attached to the sensor node reattempts to
  transmit unacknowledged packets until they are successfully received
  at the DM.
\item The transmission of measurement packets (from the measurement
  packets queue) is done in order of sampling.
\end{enumerate}
The transmit queue under this model evolves as
\begin{displaymath}
  Q_{k+1} = {\left( Q_{k} - Y_{k} \right)}^+ + A_{k}.
\end{displaymath}

The DM uses a sequential algorithm to detect the change point using $Q_1$ and
$\left( Y_k, Z_k \right), k\geq 1$. Denote by
$\mathcal{F}_t = \sigma\left( Q_1, (Y_k, Z_k), 1 \leq k \leq t \right)$, the
$\sigma$-algebra generated by the observations available at the DM up
to the end of (or the trailing edge of) slot $t$. Then, the DM uses a
sequential detection rule, an $\mathcal{F}_t$-stopping time $T$, to
raise an alarm declaring that a change has been detected. We assess
the performance of the detection rule, following Lorden's approach
\cite{lordenProceduresReactingChange1971}, in terms of the worst-case (essential supremum (ES))
average delay in the detection (ADD) of the change point, measured using the
$\text{ESADD}$ which, for a stopping rule $T$, is defined as
\begin{equation}
  \label{eq:lordens-criterion}
  \barexpectt{1}{T} \doteq \sup_{\nu \geq 1} \esssup \expectt{\nu}{ {\left( T - \nu + 1 \right)}^+ \mid Q_1, Y_{1}^{\nu-1}, Z_{1}^{\nu-1} },
\end{equation}
where $\esssup$ is the essential supremum and
$$Y_1^{\nu - 1} \coloneqq \left\{ Y_1, \ldots, Y_{\nu-1} \right\}, Z_1^{\nu-1} \coloneqq \left\{ Z_1, \ldots, Z_{\nu-1} \right\}.$$ We
aim to minimize the ESADD subject to the constraint on the average run-length to false alarm
$\text{ARL2FA} = \expectt{\infty}{T} \geq \gamma$, where $\gamma > 0$
is a large positive number. That is, we seek to find a stopping time
$T^\ast$ in the set
$\mathcal{C}_\gamma = \left\{T: \expectt{\infty}{T}\geq \gamma
\right\}$ such that
\begin{displaymath}
  \barexpectt{1}{T^\ast} = \inf_{T \in \mathcal{C}_\gamma}  \barexpectt{1}{T}
\end{displaymath}
if it exists; and if not, we would like to make the left-hand-side as
close to the right-hand-side as one wishes. The asymptotic analysis will involve $\gamma \rightarrow \infty$

\subsection{Log likelihood ratio analysis}
\label{sec:llr-analysis-inf-retr}
For a change point \(\nu\), the log-likelihood ratio of \(\mathbb{P}_{\nu}\) vs. \(\mathbb{P}_{\infty}\), based on observations at the DM up to time \(n\), is:
\[\log \frac{\mathbb{P}_\nu\left( Y_1^{n}, Z_1^{n} \mid Q_1  \right)}{\mathbb{P}_\infty\left( Y_1^n, Z_1^{n} \mid Q_1 \right)}.\]
The sequence \(Y_1^n\) for \(n = 2, 3, \ldots\) is not i.i.d.; for example, if for some \(i > 0\), \(Y_i = 0\), then \(Y_{i+1} \neq \emptyset\). Nevertheless, the channel service process is conditionally independent given a non-zero queue length process. For $i > \nu$, we define

\footnotesize
\begin{equation}
 \label{eq:define-Li}
 L_i \coloneqq \indicator{Q_i > 0} \log\frac{\mathbb{P}_0(Y_i \mid Q_i > 0)}{\mathbb{P}_\infty\left( Y_i \mid Q_i > 0 \right)}
 + \indicator{ Y_i = 1, D_i > Q_1 + S_\nu
   } \log\frac{f_1\left( Z_i \right)}{f_0\left( Z_i
   \right)}.
\end{equation}
\normalsize

\begin{discussion}
  $L_i$ is the log-likelihood ratio of the observations at the DM in the slot $i$. The first term is a function of the channel service process. The transmitter attempts a transmission in slot $k$ only when the transmit queue at the beginning of slot $k$, $Q_k$, is non-empty. The second term is a function of the received measurement. A measurement is received only when a successful transmission takes place, i.e., when $Y_k = 1$. Further, in a slot where the DM successfully receives a measurement, i.e., when $Y_k = 1$, we have $Z_i = X_{D_i}$.  \qed
\end{discussion}
The following lemma allows us to write the log-likelihood ratio as the sum of $L_i$, each of which is a function of random variables that correspond to slot $i$ alone.
\begin{lemma}
  \label{lem:log-likelihood-computation}
  With $L_i$ defined as in \cref{eq:define-Li}, the log-likelihood ratio satisfies:
  \begin{align}\label{eq:log-likelihood-computation}
  \log \frac{\mathbb{P}_\nu\left( Y_1^{n}, Z_1^{n} \mid Q_1 \right)}{\mathbb{P}_\infty\left( Y_1^n, Z_1^{n} \mid Q_1 \right)} = \sum_{i=\nu+1}^n L_i.
  \end{align}
\end{lemma}%
\begin{proof}[Proof sketch]
  The proof (provided in Appendix~\ref{appendix:Proof of lem:log-likelihood-computation}) uses the facts that the arrival process of the measurement samples is the same under both the probability laws $\mathbb{P}_\nu$ and $\mathbb{P}_\infty$, and that the received measurements at the DM are independent of the queuing process.
\end{proof}

\begin{discussion}
\label{disc:S-nu} 
$L_i$, as defined above, requires knowledge of $S_{\nu}$ for computation. But this is not available to the receiver. Even if $S_{\nu}$ were available to the receiver, for e.g., when the arrival process (along with $Q_1$) is known to both the transmitting end and the DM, $S_\nu$ depends on $\nu$, and the generalized likelihood ratio $C_n \coloneqq \max_{0 \leq \nu \leq n} \sum_{i = \nu+1}^n L_i$ does not simplify to the recursion $C_n = (C_{n-1} + L_n)^+$ rendering computation of $C_n$ difficult. However, if we set $S_{\nu} = 0$, then the corresponding (approximate) CUSUM not only admits the recursion but also does not require the knowledge of the arrival process for its computation. The assumption $S_{\nu}=0$ causes a few terms (the actual $S_{\nu}$ measurement samples) in the resulting $C_n$ to have negative drift. But this will not affect the asymptotic results for two reasons: (a) the time taken for a decision goes to $\infty$ as $\gamma \rightarrow \infty$ and a finite number of negative drift terms does not affect the rate of upward drift; (b) though the terms have negative drift, the truncation at 0 restricts their impact. Moreover, the worst-case detection delay is for the case when $\nu=0$, i.e., which restarts Page's CUSUM \cite[p.1380]{moustakidesOptimalStoppingTimes1986}, and $S_{\nu}=0$ brings us closer to such a restart.
\end{discussion}

Hence, for the detection of the change point, the DM uses the \CUSUM{} rule \cite{pageContinuousInspectionSchemes1954} with an update $L_i$, at slot $i$ that assumes $S_{\nu}=0$. Note that the detection rule is an $\mathcal{F}_t$--stopping time $T$, defined as
\begin{equation}
  \label{eq:cusum-definition} T(h) = \min\left\{ n \in \mathbb{N} :
C_n > h, C_n = {\left( C_{n-1} + L_n \right)}^+ \right\},
\end{equation}
where $h$ is the decision threshold, which is tuned for the \CUSUM{} rule to achieve the target false alarm performance.

\subsection{Asymptotic analysis}
\label{sec:asymptotic-analysis-geom-geom-one-model}

In this section, we analyze the performance and prove the optimality of the \CUSUM{} defined in \cref{eq:cusum-definition}, in the asymptotic regime as $\gamma \to \infty$. First, we will show that there exists a lower bound on the ESADD (defined in \cref{eq:lordens-criterion}) when $\gamma \to \infty$. Then, we will show that the \CUSUM{} rule defined in \cref{eq:cusum-definition} achieves the lower bound in the asymptotic regime. We will use Lai's \cite{laiInformationBoundsQuick1998} generalization of Lorden's asymptotic theory to general processes to prove these bounds.

To prove the lower bound on the asymptotic ESADD, we will need to augment the observation space of the DM. Define $\zeta_k = \left( Q_k, Y_k, Z_k \right)$; see \cref{fig:notationfigure} for the embedding of the component processes. The log likelihood ratio of $\zeta_1^n$, given the initial queue length $Q_1$, under $\mathbb{P}_\nu$ versus $\mathbb{P}_\infty$ is equal to
\begin{align*}
  \ell_{\nu, n} &= \log\frac{\mathbb{P}_\nu\left( \zeta_1^n \mid Q_1 \right)}{\mathbb{P}_\infty\left( \zeta_1^n \mid Q_1 \right)}
   = \log\frac{\mathbb{P}_\nu\left( Z_1^n, Y_1^n \mid Q_1 \right)}{\mathbb{P}_\infty\left( Z_1^n, Y_1^n \mid Q_1 \right)}\\
    & = \sum_{i=1}^n L_i,
\end{align*}
where we have used the fact that the queue length is independent of the received measurement, which implies that, given $\left(Q_1, Y_1^k,\right)$, $Q_k$ is equal under both the probability laws, due to the arrivals being equal, and hence,
$\log\frac{\mathbb{P}_\nu\left( Q_2^n \mid Q_1, Y_1^k \right)}{\mathbb{P}_\infty\left(Q_2^n \mid Q_1, Y_1^k \right)} = 0$. Then, by \cref{lem:log-likelihood-computation}, we have $\ell_{\nu, n} = \sum_{i=1}^{n}L_i$.

Next, to prove the upper and lower bounds on the ESADD of the \CUSUM{} algorithm in \cref{eq:cusum-definition}, we define the following quantity:
\begin{equation}\label{eq:basic-defn-I}
  I \coloneqq \lim_{n \to \infty} \frac{1}{n} \ell_{0, n}.
\end{equation}
The Markovity of $\zeta_k$, given the change point $\nu$, is clear from the evolution of the queue dynamics defined in \cref{sec:llr-analysis-inf-retr}. Further, under the \cref{assump:sampling-rate-sane}, the transmitter queue, $\left\{ Q_k: k\geq 1 \right\}$, is stable, and the Markov process $\{\zeta_k\}$ has a stationary distribution. Under the probability law $\mathbb{P}_0$, the stationary distribution of the Markov process $\{\zeta_k\}$ is given by $\Pi^{\left( p_1 \right)}_{\zeta}$, where the superscript denotes that this is the stationary distribution for the Markov chain with the post-change packet loss probability $p_1$ and $D_1 > Q_1 + S_{\nu}$.

Since $\left\{ \zeta_k \right\}$ is an aperiodic and recurrent Markov process, by the ergodic theorem for Markov processes \cite{durrettProbabilityTheoryExamples2019}, we have
\begin{equation}\label{eq:ergodic-limit-Li}
  \lim_{n\to\infty} \frac{1}{n} \ell_{0, n} = \lim_{n\to\infty} \frac{1}{n} \sum_{i=1}^n L_i = \expectt{\Pi_{\zeta}^{\left( p_1 \right)}}{ L_1}.
\end{equation}
Combining \cref{eq:ergodic-limit-Li,eq:basic-defn-I}, we have
\begin{align}
  \label{eq:define-I}
  I &= \expectt{\Pi_{\zeta}^{\left( p_1 \right)}}{ L_1} \nonumber \\
    &= \mathbb{E}_{\Pi_{\zeta}^{\left( p_1 \right)}}\left[\indicator{Q_1 > 0} \log\frac{\mathbb{P}_0(Y_1 \mid Q_1 > 0)}{\mathbb{P}_\infty\left( Y_1 \mid Q_1 > 0 \right)}\right. \nonumber\\
      &\qquad \qquad + \left. \indicator{ Y_1 = 1, D_1 > Q_1 + S_\nu
      } \log\frac{f_1\left( Z_1 \right)}{f_0\left( Z_1
      \right)}\right] \nonumber\\
    &= \expectt{\Pi_{\zeta}^{\left( p_1 \right)}}{\indicator{Q_1 > 0} \log\frac{\mathbb{P}_0(Y_1 \mid Q_1 > 0)}{\mathbb{P}_\infty\left( Y_1 \mid Q_1 > 0 \right)}} \nonumber\\
    &\quad+ \expectt{\Pi_{\zeta}^{\left( p_1 \right)}}{\indicator{ Y_1 = 1, Q_1 > 0} \log\frac{f_1\left( X_{D_1} \right)}{f_0\left( X_{D_1} \right)}}\nonumber\\
  \implies I &= \mathbb{P}_0\left( Q_1 > 0 \right) \left(\kldist{p_1}{p_0} + p_1\kldist{f_1}{f_0}\right),
\end{align}
where, the third equality is a consequence of $Z_i = X_{D_i}$, whenever $Y_i = 1$. Further, $\kldist{p_1}{p_0} \coloneqq \expectt{\Pi_{\zeta}^{\left(p_1\right)}}{\log \frac{\mathbb{P}_0\left( Y_k \mid Q_k > 0 \right)}{\mathbb{P}_\infty\left( Y_k \mid Q_k > 0 \right)}}$ is the Kullback-Leibler (KL) divergence between the Bernoulli distributions with parameters $p_1$ and $p_0$, i.e., $\kldist{p_1}{p_0} = \,p_1 \log \frac{p_1}{p_0} + (1-p_1)\log \frac{1-p_1}{1 -p_0}$. Finally, $\kldist{f_1}{f_0} \coloneqq \expectt{\Pi_{\zeta}^{\left(1\right)}}{\log \frac{f_1(X_j)}{f_0(X_j)}}$ is the KL divergence between the pre-change distribution $f_0$ and the post-change distribution $f_1$ of the measurement samples.

The stationary probability under the law $\mathbb{P}_0$ that the queue is non-empty, $\mathbb{P}_0\left( Q_1 > 0 \right)$, can be computed using the Little's theorem \cite{kleinrockPart1} for queues.  Under the probability law $\mathbb{P}_0$, the transmitter queue is a $\text{Geom}/\text{Geom}/1$ queue \cite{kleinrockPart1} with arrival rate $r$ and service rate $p_1$, and $\mathbb{P}_0\left( Q_1 > 0 \right) = r/p_1$. The quantity $I$ can therefore be rewritten as
\begin{equation}\label{eq:information-number}
  I = r \left( \frac{1}{p_1} \kldist{p_1}{p_0} + \kldist{f_1}{f_0} \right).
\end{equation}
The quantity $I$ is the average ``information'' provided by each arriving packet. For each measurement arrival, there are $1/p_1$ packet receipts received by the DM, each of which carry an average information of $\kldist{p_1}{p_0}$ and one measurement receipt, which bears an average information $\kldist{f_1}{f_0}$.

We will now proceed to prove bounds on asymptotic ESADD of the \CUSUM{} rule defined in \cref{eq:cusum-definition}. First, we show that there exists lower bound on ESADD (defined in \cref{eq:lordens-criterion}) for the observation sequence $\zeta_k$.
\begin{theorem}[Lower bound on \CUSUM{} ESADD]
  \label{thm:lower-bound-cusum} 
For the Markov process $\zeta_k$, as the ARL2FA $\gamma \to \infty$, we have
\begin{displaymath} 
\inf_{\left\{ \expectt{\infty}{T} \geq \gamma \right\}}\barexpectt{1}{T}  \geq \left(
I^{-1} + o(1) \right) \log\gamma.
  \end{displaymath}
\end{theorem}
\begin{proof}[Proof sketch]
  The theorem is a consequence of \cite[Theorem
  1]{laiInformationBoundsQuick1998}. See
  Appendix~\ref{appendix:Proof of thm:lower-bound-cusum} for details.
\end{proof}
The theorem, bounds the rate of growth of the ESADD as the
$\text{ARL2FA} \to \infty$. Note how the lower bound, even in this
case of a Markov process, is similar in form to that of an i.i.d.\
process in Lorden \cite{lordenProceduresReactingChange1971}. In the
i.i.d.\ case, the denominator had a clear interpretation as the KL
divergence between the pre-change and the post-change distributions.
Lai \cite[Eq. 6]{laiInformationBoundsQuick1998} provides an
interpretation for general probability distributions which we apply to
our observations $\zeta_k$ in \cref{eq:define-I}.

Next, we prove an upper bound on the asymptotic ESADD of our \CUSUM{}
algorithm (defined in \cref{eq:cusum-definition}) as the threshold
$h \to \infty$.
\begin{theorem}[Upper bound on \CUSUM{} ESADD]
  \label{thm:upper-bound-cusum}
  For the \CUSUM{} defined in \cref{eq:cusum-definition}, with a
  threshold $h$, we have $\expectt{\infty}{T} < \infty$ and
  \begin{displaymath}
    \barexpectt{1}{T} \leq \left( I^{-1} + o(1)
\right)h,\,\text{as } h\to\infty.
  \end{displaymath}
\end{theorem}
\begin{proof}[Proof sketch]
  The theorem uses an upper bound in
  \cite[Theorem. 4]{laiInformationBoundsQuick1998} on the asymptotic ESADD
  for the \CUSUM{} algorithm under certain assumptions. See
  Appendix~\ref{appendix:Proof of thm:upper-bound-cusum} for a detailed proof.
\end{proof}
Thus, we have proved that the \CUSUM{} algorithm that uses a threshold
$h = \log\gamma$, where $\gamma \to \infty$ achieves the lower bound,
and is hence an asymptotically optimal sequential detection algorithm.
For a sufficiently large threshold $h$, and $I$ defined as in
\cref{eq:define-I}, the detection rule thus has ESADD
\begin{equation}\label{eq:esadd-cusum}
  \barexpectt{1}{T} \approx \frac{h}{I} \left( 1 + o(1) \right).
\end{equation}

\subsection{Extension to periodic sampling}
\label{sec:extension-periodic-sampling}
In \cref{sec:system-model}, we assumed that the measurements are
sampled such that there is at most one sample in any channel slot, the
sampling is i.i.d.\ across slots, and the probability of a packet
arrival at the transmit queue at each slot is given by the parameter
$r$. Suppose instead that we consider that the sample measurements are
produced by a periodic sampler, with sampling interval $s =
1/r$. Then, there is a packet arrival at the transmit queue once every
$s$ slots.  We first note that since the packet arrivals are
independent of the change, \cref{lem:log-likelihood-computation} holds
for this case, and the DM uses the same \CUSUM{} algorithm defined in
\cref{eq:cusum-definition} in this case too.

The analysis in \cref{sec:asymptotic-analysis-geom-geom-one-model}
crucially uses the Markov property of $\zeta_k$. To preserve the
Markovian property in this case, we augment the state space of
observations by defining
$\zeta_k = \left( Q_k, V_k, Y_k, Z_k \right)$, where
$V_k, k\in\mathbb{N}$ counts the number of slots to the next packet
arrival. Note that given $V_1$, and the sampling interval, $V_k$ can
be computed for all $k>1$. Hence, the log-likelihood ratio of $\zeta_1^n$, given $Q_1, V_1$, under $\mathbb{P}_\nu$ versus $\mathbb{P}_\infty$ is
\begin{displaymath}
  \ell_{\nu, n} = \log \frac{\mathbb{P}_\nu\left(
\zeta_1^n \mid Q_1, V_1\right)}{\mathbb{P}_\infty\left( \zeta_1^n \mid Q_1, V_1\right)} = \log
\frac{\mathbb{P}_\nu\left( Z_1^n, Y_1^n \mid Q_1, V_1\right)}{\mathbb{P}_\infty\left( Z_1^n, Y_1^n \mid Q_1, V_1   
\right)}.
\end{displaymath}
It is easy to see that the analysis in
\cref{sec:asymptotic-analysis-geom-geom-one-model} holds for this case
after augmenting the state space of observations.

\subsection{Extension to lossy in-order reception}
\label{sec:lossy-order-recept}
In this section, we study the setting in which the transmit queue
limits the number of attempts to retransmit unacknowledged measurement
packets to $K$. After $K$ unsuccessful attempts, the packet is removed
from the transmit queue. In this case, the DM does not receive all the
measurement samples that the sensor node gets. The DM must make a
decision on whether a change has occurred, based on the samples that
it has received. We will demonstrate that our analysis framework that
we developed in \cref{sec:infin-retr-in-order-reception} can be
extended to this case.

The log-likelihood ratio of the observations can easily be shown to be
the same as developed in \cref{sec:infin-retr-in-order-reception}, and
using the fact that the packet arrivals are independent of the change,
\cref{lem:log-likelihood-computation} holds. The DM uses the \CUSUM{}
algorithm defined in \cref{eq:cusum-definition} with the \CUSUM{}
update $L_i$ at the end of the slot $i$ as defined in
\cref{eq:define-Li}. The observations
$\left\{ \zeta_k = \left( Q_k, Y_k, Z_k \right): k \geq 1 \right\}$
are Markov and hence the analysis in
\cref{sec:asymptotic-analysis-geom-geom-one-model} is applicable in
this setting. The quantity $I$ defined in \cref{eq:define-I} has the
same form in this case, with the exception of
$\mathbb{P}_0\left( Q_1 > 0 \right)$ being equal to the probability of the
transmit queue being non-empty when the maximum number of
retransmissions per measurement packet is capped at $K$. This quantity
can be computed numerically for a given $p_1$ and $K$. The asymptotic
ESADD of the \CUSUM{} algorithm continues to have the same form as in
\cref{eq:esadd-cusum}.

When number of retransmissions $K=1$, the transmitter described in
\cref{sec:system-model} uses a best-effort service and does not
attempt retransmissions of failed packet transmission. There is at
most one new measurement in a slot and there are no retransmissions. Hence, $Q_k \in \left\{ 0, 1 \right\}$ and $Q_k$ is i.i.d.\ when the packet arrival process is Bernoulli. Further, for a packet arrival rate $r$, $\mathbb{P}_0\left( Q_1 > 0 \right) = r$. \Cref{eq:define-I} then takes the form $ I = \expectt{1}{L_1} = r \left( \kldist{p_1}{p_0} +
  p_1\kldist{f_1}{f_0}\right)$. Further, the process
$\left\{ \zeta_k = \left( Q_k, Y_k, Z_k\right): k \in \N \right\}$ is i.i.d.\ before and after the change point $\nu$, and hence the usual \CUSUM{} calculations \cite{tartakovskySequentialAnalysisHypothesis2015} apply. The asymptotic ESADD of the \CUSUM{} algorithm continues to have the same form as in \cref{eq:esadd-cusum}.

If in addition to $K=1$, the channel loss process $\{Y_k\}$ has the same distribution before and after the change point, i.e., $p_0 = p_1$, it is easy to see that the log-likelihood ratio reduces to the form $\ell_{\nu, n} = \sum_{i=\nu+1}^n\indicator{Q_i > 0, Y_1 =
  1}\log\frac{f_1\left( X_{D_i} \right)}{f_0\left( X_{D_i} \right)}$ and that the \CUSUM{} update $L_i$ is zero for all lost observations. The quantity $I$ in \cref{eq:define-I} is equal to $I = \expectt{1}{L_1} = rp_1\kldist{f_1}{f_0}.$

\section{The effect of transmission queue discipline}
\label{sec:non-fcfs-discipline}
In this section, we consider various service disciplines being used in the transmitter queue, for example First-Come-First-Served (FCFS), LCFS transmit queue disciplines.  Intuitively, we can expect that the LCFS service discipline will present the DM with post-change measurements earlier than FCFS.  We assume, as in \cref{sec:infin-retr-in-order-reception}, that all packets are attempted for retransmission until success, and that the transmit queue buffer is large enough such that it never drops any packet until the packet has been successfully delivered at the DM. Hence, the DM receives all the measurements sampled by the sensor node.

Recall that $J_k$ denotes the sample index of the measurement successfully received at the DM in slot $k$. The transmit queue operates under one of the following service disciplines. 
\begin{inparaenum}[(a)]
    \item \textbf{FCFS:} If packets are successfully received in slots $i$ and $j$, then $J_i < J_j$ whenever $i < j$; that is, departures occur in the same order as arrivals.  
    \item \textbf{Non-FCFS:} The discipline is non-idling; the transmitter always attempts a packet when the queue is nonempty. Preemption is allowed, so after one packet’s transmission attempt, the next may not be the same packet. For the same arrival and channel-loss sequence, the queue length and departure processes remain identical to those under FCFS; only the order of departures differs. An example is the LCFS discipline, where the most recently arrived packet is transmitted first.
\end{inparaenum}

Our problem, as before, is to design a sequential detection algorithm for QCD using the observations
$\left\{ Y_k, Z_k, J_k: k = 1, 2, \ldots \right\}$. Denote by $\mathcal{F}_t = \sigma\left( Y_k, Z_k, J_k; 1\leq k \leq t \right)$, the $\sigma$-algebra generated by the observations available at the DM at end of slot $t$. We seek to find an $\mathcal{F}_t -$stopping time $T^\ast$ such that

\small
\begin{displaymath}
  \barexpectt{1}{T^\ast} = \inf_{T \in \mathcal{C}_\gamma}\underbrace{\sup_{\nu \geq 1} \esssup \expectt{\nu}{{\left( T - \nu + 1 \right)}^+ \mid Y_1^{\nu-1}, Z_1^{\nu-1}, J_1^{\nu-1}}}_{\barexpectt{1}{T}},
\end{displaymath}
\normalsize
where $\mathcal{C}_\gamma = \left\{ T: \expectt{\infty}{T} \geq \gamma \right\}$ is the set of all $\mathcal{F}_t$--stopping times with a $\text{ARL2FA} \geq \gamma$. Note that $J_1^k$ represents the sequence of sampling indices of all successfully received measurements up to slot $k$.

\subsection{Log likelihood ratio analysis}
\label{sec:fifo-lifo-log-likelihood-ratio-analysis}
We assume that the transmitter queue length at the beginning of the
first slot, $Q_1$, is known a priori to the DM. The packets already in
the transmitter queue, before the QCD process starts, will need to be
discarded by the DM. The log-likelihood ratio of $\mathbb{P}_\nu$
vs. $\mathbb{P}_\infty$, based on observations at the DM at time $n$, is
\begin{displaymath}
 \ell_{\nu, n} = \log \frac{\mathbb{P}_\nu\left( Y_1^{n}, Z_1^{n}, J_1^{n}\mid Q_1
  \right)}{\mathbb{P}_\infty\left( Y_1^n, Z_1^{n}, J_1^{n} \mid Q_1 \right)}. 
\end{displaymath}
\begin{lemma}
  \label{lem:non-fifo-log-likelihood-ratio}
  Suppose that by time $n$, the samples corresponding to the indices
  $k_1, k_2, \ldots, k_{D_n}$ are received at the DM, then
  \begin{align}
    \label{eq:log-likelihood-ratio}
    \ell_{\nu,n} &= \log\frac{\mathbb{P}_\nu\left( Y_1^n \mid Q_1 \right)}{\mathbb{P}_\infty\left( Y_1^n \mid Q_1 \right)} + \log\frac{\mathbb{P}_\nu\left( Z_1^n \mid Q_1, Y_1^n, J_1^n \right)}{\mathbb{P}_\infty\left( Z_1^n \mid Q_1, Y_1^n, J_1^n \right)} \nonumber \\
    &= \log\frac{\mathbb{P}_\nu\left( Y_1^n \mid Q_1 \right)}{\mathbb{P}_\infty\left( Y_1^n \mid Q_1 \right)} + \log\frac{\mathbb{P}_\nu\left( X_{k_1}, \ldots, X_{k_{D_n}}\right)}{\mathbb{P}_\infty\left( X_{k_1}, \ldots, X_{k_{D_n}} \right)}.
  \end{align}
\end{lemma}
\begin{proof}[Proof sketch]
The proof uses the fact that the arrival processes under both the transmit queue disciplines are the same, and follows the steps of the proof of \cref{lem:log-likelihood-computation}. See Appendix~\ref{appendix:Proof of lem:non-fifo-log-likelihood-ratio} for a detailed proof.
\end{proof}

The DM rearranges the measurements $\left( Z_1, Z_2, \ldots, Z_k \right)$ in an increasing order of arrival index as
$\left(Z_k^{\left( 1 \right)}, Z_k^{\left( 2 \right)}, \ldots,
  Z_k^{\left( k \right)} \right)$. At the end of each time slot $k$, for
$\nu < i \leq k$, define
\begin{align}
  \label{eq:non-fifo-llr-defn}
  L_i\left( k \right) &= \indicator{ Q_i > 0 }\left( \log \frac{\mathbb{P}_\nu\left( Y_i \mid Q_i > 0 \right)}{\mathbb{P}_\infty\left( Y_i \mid Q_i > 0\right)} \right. \nonumber\\
  &\qquad + \left. \indicator{Y_i = 1, D_i > Q_1 + S_\nu }\log\frac{f_1\left( Z_{k}^{\left( i \right)} \right)}{f_0\left( Z_k^{\left( i \right)} \right)} \right).
\end{align}
From \cref{lem:non-fifo-log-likelihood-ratio} and the proof of
\cref{lem:log-likelihood-computation}, it is clear
that,$\ell_{\nu, k} = \sum_{k = \nu+1}^k L_i\left( k \right)$. Again, following Discussion~\ref{disc:S-nu}, we take $S_{\nu}=0$. But it is
important to note that the terms $L_i\left( k \right)$ must be
recomputed for all $i \leq k$ at the end of each time slot $k$. The
\CUSUM{} rule, in general, cannot be written as a recursion when the
DM receives measurement samples out-of-order.
The DM uses the detection rule, an $\mathcal{F}_t$-stopping
time $T$, defined as
\begin{equation}
  \label{eq:non-fifo-cusum-definition}
  T(h) = \min\left\{ n\in \N: C_n > h, C_n = \max_{1\leq j \leq n} \ell_{j, n} \right\},
\end{equation}
where the threshold $h$ is tuned to achieve the target
$\text{ARL2FA} = \gamma$.

\begin{discussion}
  The indices $J_1^k$ are contiguous integers and in a
  strictly increasing order if a FCFS queuing discipline is used at
  the transmit queue. If a different queuing discipline is used at the
  transmit queue, then $J_1^k$ is not a sequence of increasing
  integers. When the FCFS transmit discipline is used, then the
  $i$\textsuperscript{th} ordered sample at the end of slot $k$,
  $Z_k^{\left( i \right)} = Z_i$, for each $i \leq k$ and each
  $k\in \N$.  Further, when the FCFS transmit discipline is used at the transmit queue, \cref{eq:non-fifo-llr-defn} reduces to
  \cref{eq:define-Li}, and \cref{eq:non-fifo-cusum-definition} reduces to \cref{eq:cusum-definition} and hence the CUSUM statistic can be computed recursively.
\end{discussion}

\begin{discussion}
The false alarm performance of the \CUSUM{} algorithm in
\cref{eq:non-fifo-cusum-definition} is measured by the ARL2FA
$\expectt{\infty}{T(h)}$. Under the probability law $\mathbb{P}_\infty$, note
that at the end of each time slot $k$,
\begin{displaymath}
\left( Z_1, Z_2, \ldots, Z_i \right) \,{\buildrel d \over =}\, \left(
  Z_k^{\left( 1 \right)}, Z_k^{\left( 2 \right)}, \ldots, Z_k^{\left(
      i \right)} \right)  \forall i\leq k,
\end{displaymath}
where $\,{\buildrel d \over =}\,$ is used to denote equality in
distribution. This is because, while the departure instants of the
measurements may differ depending upon the queue discipline in use at
the transmitter queue, all the measurement samples are generated from
the pre-change distribution. Hence, the properties of the stopping
time remain unchanged under the $\mathbb{P}_\infty$ probability law for any
transmit queue discipline used. Consequently, the threshold $h$ to achieve a target ARL2FA is the same for all queue disciplines.  \qed
\end{discussion}

In the next section, we analyze the effect of 
transmit queue disciplines on the detection performance of the \CUSUM{}
algorithm defined in \cref{eq:non-fifo-cusum-definition}, under the
assumption that the pre-change and post-change distributions are
likelihood ratio ordered.

\subsection{Analysis of detection performance under a condition of likelihood ratio ordering}
\label{sec:non-asympt-analys-non-fcfs}

We will require the notion of likelihood ratio ordering for the
analysis of the detection performance under different queuing
disciplines.
\begin{definition}[LR ordering \cite{shakedStochasticOrders2007}]
  \label{defn:lr-ordering}
  Suppose that $g\left( \cdot \right), f\left( \cdot \right)$ are
  probability density functions. Then, the likelihood ratio ordering
  $f \lrleq g$ is said to exist if for every $x \leq y$ in the union
  of the support of $f\left( \cdot \right)$ and
  $g\left( \cdot \right)$,
  \[
    g(x) f(y) \leq f(x) g(y).
  \]
\end{definition}
The following order will also be of relevance.
\begin{definition}[Stochastic ordering \cite{shakedStochasticOrders2007}]
  We say that the stochastic ordering $X \stochleq Y$ exists between
  two random variables $X$ and $Y$ if for all
  $x \in \left( -\infty, \infty \right)$, we have
  \begin{displaymath}
    P\left( X > x \right) \leq P\left( Y > x \right).
  \end{displaymath}
  In words, we say that $X$ is smaller than $Y$ in the stochastic
  order.\qed
\end{definition}
LR ordering implies stochastic ordering \cite[Theorem.~1.C.1]{shakedStochasticOrders2007}, and thus LR ordering is stronger.

\begin{assumption}\label{assump:lr-ordering-of-dists}
We assume that the pre-change and post-change distributions $f_0$ and $f_1$ are likelihood ratio ordered, i.e.,
\begin{equation}\label{eq:likelihood-ratio-ordering}
  f_0\underset{LR}{\leq}f_1.
\end{equation}
\end{assumption}

For example, if $f_0 = \mathcal{N}(0,1)$ and $f_1 = \mathcal{N}(\mu,1)$ with $\mu > 0$, then $f_0 \leq_{LR} f_1$ since $\frac{f_1(x)}{f_0(x)} = \exp(\mu x - \mu^2/2)$ increases monotonically with $x$, and Assumption~\ref{assump:lr-ordering-of-dists} is satisfied.

Under \cref{assump:lr-ordering-of-dists}, the following lemma follows.
\begin{lemma}
  \label{lem:lr-ordering-of-samples}
  For $i, j\in \mathbb{N}, i \leq j$, the samples corresponding to the
  sampling indices $i, j$ have the following likelihood ratio
  ordering:
  \begin{displaymath}
    X_i \lrleq X_j.
  \end{displaymath}
\end{lemma}
\begin{proof}
  Suppose that $s_{\nu}$ corresponds to the first sample generated
  after the change point $\nu$. Then, we have the following three
  cases:\\
  1) $i\leq j < s_{\nu}$. Here, $X_i, X_j \sim f_0$, hence $X_i \lrleq X_j$.\\
  2) $i < s_{\nu} \leq j$. Here, $X_i \sim f_0$ and $X_j \sim f_1$. Hence, $X_i \lrleq X_j$.\\
  3) $s_{\nu} \leq i \leq j$. In this case, $X_i, X_j \sim f_1$, hence $X_i \lrleq X_j$.
\end{proof}
Since there is no residual packet loss on the transmit link (i.e., all packets are reattempted until successful), every packet is eventually delivered to the DM, albeit with stochastic delay. At the end of each busy period, both FCFS and any alternate non-idling discipline deliver the same set of packets to the DM. Consequently, the log-likelihood ratio of received observations is identical across disciplines at busy-period completion. However, within a busy period, the packets received up to any slot may differ across disciplines. Hence, we can separate our analysis into (i) completed busy periods, and (ii) the ongoing busy period. In the following discussion we consider slots within a busy period.

Consider a busy period during which $n$ packets are delivered to the DM, indexed $1,2,\ldots,n$. Let $k_1, k_2, \ldots, k_n$ denote the sampling indices of packets received in order of reception; this sequence is just a permutation of $(1,2,\ldots,n)$. At some intermediate slot $i<n$ (i.e., before the busy period ends), the DM under a non-FCFS discipline has received samples with indices $k_1,\ldots,k_i$, whereas under FCFS it would have received $1,\ldots,i$. Suppose that $k_1, \ldots, k_i$ are rearranged in ascending order as $k_i^{\left( 1 \right)}, \ldots, k_i^{\left( i \right)}$. Here $k_i^{(j)}$ denotes the $j$-th smallest sampling index among the first $i$ received packets. Then, the following holds:
\begin{lemma}
  \label{lem:indices-ordering}
  For all $i \leq n$, $k_i^{\left( j \right)} \geq j$, for
  $j = 1, 2, \ldots, i$.
\end{lemma}
This follows immediately on making the following observation:
\begin{lemma}
  \label{lem:max-of-permuted-set}
  For any set of unique integers $S\subseteq \left\{ 1, 2, \ldots, n \right\}$, such that $\left| S \right| = k$,
  \begin{displaymath}
    \max_{i \in S} i \geq k
  \end{displaymath}
\end{lemma}
\begin{proof}
  Suppose not. Let $j = \max_{i \in S} i < k$, then $\left| S \right| < k$, which is a contradiction to the hypothesis that $\left| S \right| = k$.
\end{proof}
\begin{proof}[Proof of \cref{lem:indices-ordering}]
  Define the set $S_i = \left\{k_1, \ldots, k_i\right\}$, the set of the indices of the received samples at the DM at the end of the $i$\textsuperscript{th} slot in the busy period. Now, use \cref{lem:max-of-permuted-set} to obtain $\max_{j \in S_i} j = k_i^{\left(i\right)} \geq i$. Recursively apply \cref{lem:max-of-permuted-set} for each $j = 1, 2, \ldots, i-1$ by putting $S_i^{\left( i-j \right)} = S_i \setminus \left\{ k_i^{\left( i-j+1\right)} \right\}$ to obtain $\max_{j\in S_i^{\left( i-j \right)}}j = k_i^{\left( i-j \right)} \geq i - j$.
\end{proof}
Suppose that in the busy period during which $n$ packets are communicated to the DM, under the LCFS transmit queue discipline, the sample indices of the measurements received at the DM at the end of the reception slot $i < n$ are denoted by $\lambda_1, \lambda_2, \ldots, \lambda_i$. Suppose that the $\lambda_1, \lambda_2, \ldots, \lambda_i$ are rearranged in ascending order as $\lambda_i^{\left( 1 \right)}, \lambda_i^{\left( 2 \right)}, \ldots, \lambda_i^{\left( i \right)}$. While Lemma~ \ref{lem:indices-ordering} describes a general ordering property for any non-FCFS discipline, the following result specifically is for the LCFS policy, showing that LCFS always serves newer samples (larger sampling indices) sooner.
\begin{lemma}\label{lem:indices-ordering-lcfs}
  The LCFS queue discipline satisfies the following: for each $j \leq i, i\leq n$,
\begin{displaymath}
    \lambda_i^{\left( j \right)} \geq k_i^{\left( j \right)}.
\end{displaymath}
\end{lemma}
\begin{proof}
  This holds because the LCFS transmit queue discipline, at each slot, transmits the packet corresponding to the measurement with the largest sampling index.
\end{proof}
The following lemma states a likelihood ratio ordering on the
rearranged observations under the different queuing disciplines.
\begin{lemma}
  \label{lem:x-lr-ordering}
  For each $j \leq i$, $i \leq n$,
  \begin{displaymath}
    X_j  \lrleq X_{k_i^{\left(j\right)}} \lrleq X_{\lambda_i^{\left( j \right)}}. 
  \end{displaymath}
  \hfill\(\Box\)
\end{lemma}
\begin{proof}
This is a simple consequence of the LR ordering of the
distributions $f_0, f_1$, and
\cref{lem:lr-ordering-of-samples,lem:indices-ordering,lem:indices-ordering-lcfs}.
\end{proof}
Denote by $C_m^{(F)}$, $C_m^{(A)}$, and $C_m^{(L)}$ the \CUSUM{} statistics computed at the DM at the end of slot $m$ under FCFS, an alternate non-FCFS, and LCFS queue disciplines, respectively.
\begin{proposition}
  \label{prop:non-fcfs-cusums-lr-ordered}
  Under \cref{assump:lr-ordering-of-dists},
  \begin{displaymath}
    C_m^{\left( F \right)} \leq_{st} C_m^{\left( A \right)} \leq_{st} C_m^{\left( L \right)}, \quad m \geq 1.
  \end{displaymath}
\end{proposition}
\begin{proof}[Proof sketch]
Since, the \CUSUM{} is a monotone function of the likelihoods of the measurement samples, the proof uses \cref{lem:x-lr-ordering} and the property that monotone functions of LR ordered random variables are LR ordered, to then show that the \CUSUM{} statistics are stochastically ordered. A detailed proof which involves induction is provided in Appendix~\ref{appendix:Proof of prop:non-fcfs-cusums-lr-ordered}
\end{proof}
We now conjecture that the stopping times associated with the LCFS on the one hand, the FCFS on the other hand, and any other queuing discipline are then stochastic ordered as follows.
\begin{conjecture}
  \label{thm:cusums-are-st-ordered}
For a threshold $h_\gamma$, tuned to achieve $\text{ARL2FA} = \gamma$, we conjecture $N^{\left( L \right)}(h_\gamma) \stochleq N^{\left( A \right)}(h_\gamma) \stochleq N^{\left( F \right)}(h_\gamma).$
\qed
\end{conjecture}
See Appendix~\ref{appendix:disc of thm:cusums-are-st-ordered} for the proof challenges.

\begin{discussion}[Asymptotic ESADD]
Irrespective of the transmit queue discipline, once a busy period completes, the DM eventually receives all generated measurements since there is no residual packet loss. Hence, the rearranged measurement sequence and therefore the \CUSUM{} statistic and stopping time are identical across disciplines at the end of each busy period. The difference in stopping times under FCFS and any non-FCFS policy is thus bounded by the finite busy-period length. Consequently, as the target ARL2FA grows large, the asymptotic ESADD of the \CUSUM{} test is identical across all non-idling queue disciplines. In summary, the asymptotic performance of
CUSUM is unaffected by queuing discipline. Numerical simulations highlight the benefit of further heuristic policies for finite ARL2FA (see \cref{sec:practical-considerations-discussion}). \qed
\end{discussion}

\section{Multi-Sensor: Homogeneous Case}
\label{sec:multi-sens}

Building on the single-sensor system model presented in \cref{sec:system-model}, we now extend the framework to a multi-sensor setup. The system consists of $L$ sensors, collectively denoted by the set $\mathcal{L}$. System slots are synchronous across all the sensors. Each sensor independently samples a common process using an asynchronous Bernoulli sampling mechanism, where each sensor samples with a probability $r$ per time slot ($0 < r < 1$).

Let $\{\bm{X}_k\}_{k \geq 0}$ represent the sequence of measurements sampled at time slot $k$, where $\bm{X}_k = \left( X_{k,1}, X_{k,2}, \ldots, X_{k,L} \right)$. Here, $X_{k,l}$ denotes the measurement made by sensor $l \in \mathcal{L}$ in slot $k$. The measurements are independent across sensors and time.

For each sensor $l \in \mathcal{L}$, the measurements are distributed as:
\begin{align*}
X_{k,l} \sim
\begin{cases}
f_0 & \text{if } k \leq  \nu, \\
f_1 & \text{if } k > \nu.
\end{cases}
\end{align*}
Each sensor maintains a transmission queue. In every time slot, a scheduling policy selects one sensor with a non-empty queue to transmit its measurement to the decision maker (DM) over the shared wireless channel. This transmission process has two components: (i) queue selection: one sensor is chosen uniformly at random from the set of non-empty queues, modeling an ideal random-access mechanism without contention overheads; and (ii) packet transmission: once selected, the packet experiences either success or loss over the channel. The resulting channel service process $Y_k$, as defined in Section~\ref{sec:system-model}, takes values in $\{1,0,\emptyset\}$, corresponding to successful, failed, or absent transmission, respectively.

We assume that all sensor–DM links have identical channel statistics, corresponding to equal link lengths and i.i.d.\ fading across sensors. Since only one sensor transmits in any slot, a single process $Y_k$ suffices to represent the observed channel behavior with common packet-loss parameters $(p_0,p_1)$. The probability of a successful transmission is then given by
\begin{align*}
\mathbb{P}(Y_k = 1 \mid \|\bm{Q}_k\|_0 > 0) =
\begin{cases}
p_0 & \text{if } k \leq \nu,\\
p_1 & \text{if } k > \nu.
\end{cases}
\end{align*}
Here, $\bm{Q}_k = (Q_{k,1}, Q_{k,2}, \ldots, Q_{k,L})$ denotes the queue-length vector at the start of slot $k$, and $\|\bm{Q}_k\|_0 > 0$ indicates that at least one sensor has a non-empty queue at the start of slot $k$. To ensure queue stability, we require $\sum_i r_i < \min\{p_0,p_1\}$.

The DM observes a quadruple $(U_k, Y_k, J_k, Z_k)$ at the end of each time slot $k$ defined as:
\begin{align*}
(U_k, Y_k, J_k, Z_k) =
\begin{cases}
(u_k, 1, j_k, X_{j_k,u_k}), \ \text{on successful transmission,} \\
(\ast, 0, \ast, \ast), \ \text{on unsuccessful transmission,} \\
(\ast, \emptyset, \ast, \ast ), \ \text{on no transmission.}
\end{cases}
\end{align*}
Here, $u_k \in \mathcal{L}$ denotes the index of the sensor whose measurement was successfully received, and $j_k$ is the sampling time of the received measurement. The placeholders ($\ast$) indicate unobserved values in the cases of transmission failure or an empty queue.

\subsection{Analysis}\label{subsec:MS-analysis}
Let the change point be $\nu$, and define the log-likelihood ratio as:
\[ l_{\nu,n} = \log \frac{\mathbb{P}_{\nu}(U_1^n, Y_1^n, J_1^n, Z_1^n \mid \bm{Q}_1)}{\mathbb{P}_{\infty}(U_1^n, Y_1^n, J_1^n, Z_1^n \mid \bm{Q}_1)} \]

\begin{proposition}
  \label{prop:MS-log-likelihood-computation}
  The log-likelihood ratio, $l_{\nu,n}$ satisfies:
  \begin{align}\label{eq:MS-log-likelihood-computation}
  l_{\nu,n} &= \log \frac{\mathbb{P}_{\nu}(U_1^n, Y_1^n, J_1^n, Z_1^n \mid \bm{Q}_1)}{\mathbb{P}_{\infty}(U_1^n, Y_1^n, J_1^n, Z_1^n \mid \bm{Q}_1)} \nonumber\\
  &= \sum_{k=1}^{n} \bigg( \mathbbm{1}_{\{\|\bm{Q}_k\|_0 > 0\}} \log \frac{\mathbb{P}_{\nu}(Y_k \mid \|\bm{Q}_k\|_0 > 0)}{\mathbb{P}_{\infty}(Y_k \mid \|\bm{Q}_k\|_0 > 0)}  \nonumber\\
  &\qquad\qquad\qquad\quad+  \mathbbm{1}_{\{Y_k =1,  J_k > Q_{1, U_k} + S_{\nu, U_k}}\} \log \frac{f_1(X_{J_k,U_k})}{f_0(X_{J_k,U_k})} \bigg)
  \end{align}
\end{proposition}
\begin{proof}
See Appendix~\ref{appendix:Proof of prop:MS-log-likelihood-computation}.
\end{proof}

Again, based on discussion~\ref{disc:S-nu}, we set $S_{\nu, U_k} = 0$. Denote by \(\mathcal{F}_t = \sigma\left( U_k, Y_k, J_k, Z_k; 1\leq k \leq t \right)\), the $\sigma$-algebra generated by the observations available at the DM up to the end of (or the trailing edge of) slot $t$. We seek a stopping time that solves

\small
\begin{align}\label{eq:MS-ESADD}
\inf_{T \in \mathcal{C}_{\gamma}} \underbrace{\sup_{\nu \geq 1} \esssup \mathbb{E}_{\nu} [(T-\nu+1)^+|U_1^{\nu-1}, Y_1^{\nu-1}, J_1^{\nu-1}, Z_1^{\nu-1}]}_{\doteq \bar{\mathbb{E}}_1[T]}
\end{align}
\normalsize
where \( \mathcal{C}_{\gamma} = \{T: \mathbb{E}_{\infty}[T] \geq \gamma\} \) and $\mathbb{E}_{\infty}[T]$ is ARL2FA. As before, the DM uses the \CUSUM{} stopping rule:
\begin{align}\label{eq:MS-stopping_time}
T(h) = \min \{n \in \mathbb{N} : \underset{1 \leq j \leq n}{\max} l_{j,n} > h\},
\end{align}
where $h$ is a threshold tuned to achieve the desired false alarm performance.

\subsection{Performance}\label{subsec:MS-performance}
In this section, we analyze the detection performance of the proposed \CUSUM{} variant in \cref{eq:MS-stopping_time} in the asymptotic regime where $\mathbb{E}_{\infty}[T] \to \infty$. Similar to the single-sensor case in \cref{eq:basic-defn-I}, we begin by defining the average information:

\small
\begin{align}\label{eq:MS-kl term}
    \frac{1}{n} l_{0,n}
     = &\sum_{k=1}^{n} \bigg( \mathbbm{1}_{\{\|\bm{Q}_k\|_0> 0\}} \log \frac{\mathbb{P}_{\nu}(Y_k| \|\bm{Q}_k\|_0> 0)}{\mathbb{P}_{\infty}(Y_k| \|\bm{Q}_k\|_0> 0)} \nonumber\\
    &  \quad \quad
    +  \mathbbm{1}_{\{Y_k =1, J_k > Q_{1,U_k}\}} \log \frac{f_1(X_{J_k,U_k})}{f_0(X_{J_k,U_k})} \bigg)
\end{align}
\normalsize
To analyze the asymptotic behavior of this quantity, we define the augmented observation space 
$\bm{\zeta}_k = (\bm{Q}_k, U_k, Y_k, J_k, Z_k)$. Let $\bm{\zeta}_1^n = (\bm{\zeta}_1, \bm{\zeta}_2, \ldots, \bm{\zeta}_n)$ denote the sequence of augmented observations from slot $1$ to $n$. The log-likelihood of $\bm{\zeta}_{1}^{n}$, given the initial queue lengths $\bm{Q}_1$, under $\mathbb{P}_{\nu}$ versus $\mathbb{P}_{\infty}$, is defined as:

\small
\begin{align} \label{eq:MS-zeta-llr}
l'_{\nu,n} &\coloneq \log \frac{\mathbb{P}_{\nu}(\bm{\zeta}_1^n|\bm{Q}_1)}{\mathbb{P}_{\infty}(\bm{\zeta}_1^n|\bm{Q}_1)} \nonumber\\
&= \log \frac{\mathbb{P}_{\nu}(U_1^n, Y_1^n, J_1^n, Z_1^n|\bm{Q}_1)}{\mathbb{P}_{\infty}(U_1^n, Y_1^n, J_1^n, Z_1^n|\bm{Q}_1)}+ \log \frac{\mathbb{P}_{\nu}(\bm{Q}_2^n|U_1^n, Y_1^n, J_1^n, \bm{Q}_1)}{\mathbb{P}_{\infty}(\bm{Q}_2^n|U_1^n, Y_1^n, J_1^n, \bm{Q}_1)} \nonumber\\
&= l_{\nu,n}
\end{align} 
\normalsize
where the second term vanishes because the queue evolution is independent of the change point, given $(\bm{Q}_1, U_1^n, Y_1^n, J_1^n)$.
The second equality in \cref{eq:MS-zeta-llr} arises from the fact that the queue length is independent of the measurement received at the DM.

Given the change instant, $\{\bm{\zeta}_k\}$ evolves as a positive recurrent, aperiodic, and irreducible Markov process with a stationary distribution. This follows from the fact that the queue evolution follows a stable birth-death process. Under the probability law $\mathbb{P}_0$, the stationary distribution of the Markov process $\bm{\zeta}_k$ is given by $\Pi_{\zeta}^{(p_1)}$, where the superscript denotes that this is the stationary distribution under the post-change packet loss probability $p_1$. Using the ergodic theorem \cite{durrettProbabilityTheoryExamples2019} for Markov processes, we obtain the following result for the asymptotic limit of~(\ref{eq:MS-kl term}) (i.e., $\lim_{n\to\infty}\frac{1}{n}l_{0,n}$):
\begin{proposition}[Asymptotic Information Rate in the Multi-Sensor Setting]
\label{prop:MS-I_bar}
The asymptotic information rate per unit time per unit sensor, i.e.,~(\ref{eq:MS-kl term}) normalized per sensor, is almost surely given by
\small
\begin{align}\label{eq:MS-I_bar}
    \bar{I}  &= \lim_{n\to\infty}\frac{1}{n}l_{0,n} \nonumber\\
    &= \mathbb{E}_{\Pi_{\zeta}^{(p_1)}} \bigg[ \mathbbm{1}_{\{\|\bm{Q}\|_0> 0\}} \log \frac{\mathbb{P}_{\nu}(Y| \|\bm{Q}\|_0> 0)}{\mathbb{P}_{\infty}(Y| \|\bm{Q}\|_0> 0)}+\mathbbm{1}_{\{Y =1\}}\log \frac{f_1(X_{J,U})}{f_0(X_{J,U})} \bigg] \nonumber \\
    &=\mathbb{P}(\|\bm{Q}\|_0 >0) I(p_1,p_0)  + \mathbb{P}(Y=1) I(f_1,f_0) 
\end{align}
\normalsize
where $I(f_1, f_0)$ is the KL-divergence between the pre- and post-change distributions of the measurements.
\end{proposition}

\begin{remark}
The expression for $\bar{I}$ captures two sources of information contributing to change detection in the multi-sensor network. The first term, $\mathbb{P}(\|\bm{Q}\|_0 >0) I(p_1,p_0)$, represents the information gained from changes in the channel success probability. The second term, $\mathbb{P}(Y=1) I(f_1,f_0)$, accounts for the information obtained directly from the measurement content when a sample is successfully transmitted and received. This result extends the single-sensor asymptotic analysis as seen in \cref{eq:information-number}. 
\end{remark}

\subsection{Channel Loss Independent of Change}
\label{subsec:MS-indep-channel}
In this subsection, we analyze the case where the channel loss process is independent of the change point. Specifically, we assume a constant channel loss probability, i.e., \(p_0=p_1 = p\). In this scenario, the log-likelihood ratio in \cref{eq:MS-log-likelihood-computation} simplifies to: 
\begin{align}
l_{\nu,n} = \sum_{k=1}^{n} \mathbbm{1}_{\{J_k > Q_{1,U_k} + S_{\nu, U_k}, Y_k = 1\}} \log \frac{f_1(X_{J_k,U_k})}{f_0(X_{J_k,U_k})},
\end{align}
where the conditional distribution of $Y_k$ is identical under both probability laws. The information measure, $\bar{I}$ in \cref{eq:MS-I_bar} reduces to:
\begin{align}\label{eq:MS-I_bar_indp_chan_loss}
\bar{I} = \mathbb{P}(Y = 1) I(f_1, f_0).
\end{align} 
The probability of a successful transmission, $\mathbb{P}(Y=1)$, can be derived by modeling a super queue at the transmitter, whose length equals the sum of the individual sensor queue lengths. The aggregate arrival process follows a Binomial distribution, $\text{Bin}(L,r)$, while the departure process is $\text{Bernoulli}(p)$. The event that at least one sensor queue is non-empty corresponds to the super queue being non-empty, i.e., $P(\|\bm{Q}_k\|_0 > 0) = P(L' > 0)$, where $L' = \sum_{i=1}^L Q_{k,i}$. Therefore, \[P(\|\bm{Q}_k\|_0 > 0) = \rho = \frac{Lr}{p},\] 
\noindent Now, for a G/G/1 queue, the probability that the system is empty is $1-\rho$ \cite[Chapter 7]{medhi2002stochastic} where $\rho = \frac{Lr}{p}$. Therefore, we have:
\begin{align*}
\mathbb{P}(Y = 1) = p \cdot P(\|\bm{Q}_k\|_0 > 0) = p \cdot \frac{Lr}p= Lr.
\end{align*}
This follows stability assumption if there is no buffer loss. Thus, the expression for $\bar{I}$ in \eqref{eq:MS-I_bar_indp_chan_loss} becomes:
\begin{align}
\bar{I} = Lr \cdot I(f_1, f_0).
\end{align}
Since the channel success probability is constant and does not carry information about the change, each detection-relevant update is contributed solely by the content of the received measurement. Assuming a stable system and no buffer loss, the rate at which such informative samples arrive at the decision maker is equal to the total effective sampling rate across all sensors, i.e., $Lr$. Each such sample contributes $I(f_1, f_0)$ amount of information, leading to the overall expression of $\bar{I}$. Bounds on the ESADD, defined in \cref{eq:MS-ESADD}, can be computed following a similar approach as in \cref{thm:lower-bound-cusum} and \cref{thm:upper-bound-cusum}.

\section{Non-homogeneous sensors and sampling}\label{sec:MS-non-homog-sens}

In this section, we analyze the case where the sensor observations and sampling rates are non-homogeneous. Each sensor \(i\) has a sampling rate \(r_i\), and its observations follow distinct distributions, \(f_{0,i}\) and \(f_{1,i}\), before and after the change, respectively. We also let the channel service process to be independent of the change instant as in \cref{subsec:MS-indep-channel}. The log-likelihood ratio in \cref{eq:MS-log-likelihood-computation} then becomes:
\[
l_{\nu,n} = \sum_{k=1}^{n} \mathbbm{1}_{\{J_k > Q_{1,U_k} + S_{\nu, U_k}, Y_k = 1\}} \log \frac{f_{1,U_k}(X_{J_k,U_k})}{f_{0,U_k}(X_{J_k,U_k})}.
\]
\begin{proposition}
  \label{prop:MS-I-non-homogeneous}
  The quantity $\bar{I}$ for the non-homogeneous setting can be written as:
  \begin{align}
    \bar{I} = \sum_{i=1}^{L} r_i I(f_{1,i}, f_{0,i})
  \end{align}
\end{proposition}
\begin{proof}
See Appendix~\ref{appendix:Proof of prop:MS-I-non-homogeneous}.
\end{proof}
This expression highlights that the overall information rate is determined by the sum of the individual sensors’ sampling rates, each weighted by the corresponding KL divergence between the pre- and post-change distributions.

\begin{discussion}
Each sensor samples observations with probability \(r_i\) in a given time slot. Once an observation is sampled, it is placed in the respective sensor queue, awaiting transmission via a wireless channel with a success probability of \(p\). The receiver computes the \CUSUM{} statistic upon receiving an observation successfully. Optimally scheduling sensor measurements to the DM is challenging because the sensors are usually non-homogeneous: each sensor may have a different noise variance, leading to varying levels of information content about the change. Specifically, sensors with larger KL divergence between pre- and post-change observation distributions provide more information. As a result, samples from such sensors are more valuable for detection, but may spend more time waiting in the queue, while samples from higher-variance (lower-information) sensors may be relatively fresh. The scheduling policy must balance this trade-off between the information content and the freshness of observations to minimize detection delay while keeping the false alarm rate low. We now explore some heuristic solutions in the following subsections.
\end{discussion} 

\subsection{Discounted Information Scheduling Policy - A Heuristic}\label{subsec:dis_info_sch_pol}

To arrive at a scheduling policy for sensor samples, we use an interchange argument \cite[Chapter 4]{bertsekas2012dynamic}. Let the associated information content of the sensor sampled at slot $k$ be denoted as $I^\ast_{U_k}$.

At the beginning of each slot, the transmitter may have several waiting observations from different sensors, each of which has already experienced some delay and may carry different levels of information. In this non-homogeneous setting, the scheduling policy must account for both the information content of each measurement and the additional delay that would be incurred by not transmitting it immediately.

The goal is to decide which sensor's sample to transmit at each time step. Let \(V\) be the optimal scheduling policy, which selects a sample \(X_{J_k,U_k}\) at time \(k\):
\[
V = (X_{J_0,U_0}, \ldots, X_{J_{k-1},U_{k-1}}, X_{j,i}, X_{j',i'}, \ldots, X_{J_{N-1},U_{N-1}}).
\]
where $N$ is the total number of samples available at the transmitter. Now, consider an alternative policy \(V'\) that interchanges the order of two samples from different sensors:
\[
V' = (X_{J_0,U_0}, \ldots, X_{J_{k-1},U_{k-1}}, X_{j',i'}, X_{j,i}, \ldots, X_{J_{N-1},U_{N-1}}).
\]

To capture the trade-off between information and timeliness, we define a cost structure as:
\[
C_{U_k} = \alpha^{k - J_k} I^\ast_{U_k},
\]
where \(0 < \alpha < 1\) is a discount factor that reduces the benefit of transmitting an older sample ($J_k < k$). Intuitively, if $k-J_k$ is large, then $\alpha^{k-J_k}$ becomes small, yielding a “cost” that can be lower—but we will clarify how we use this to rank samples.

With policy \( V \), the total cost can be written as:
\begin{align*}
    C(V) &= C\bigg(X_{J_0,U_0}, \ldots, X_{J_{k-1},U_{k-1}}\bigg) + C_i + \alpha C_{i'} \\
    &+C\bigg(X_{J_{k+2},U_{k+2}}, \ldots, X_{J_{N-1},U_{N-1}}\bigg)
\end{align*}
Similarly, for \( V' \), the total cost is:
\begin{align*}
    C(V') &= C\bigg(X_{J_0,U_0}, \ldots, X_{J_{k-1},U_{k-1}}\bigg) + C_{i'} + \alpha C_i \\
    &+ C\bigg(X_{J_{k+2},U_{k+2}}, \ldots, X_{J_{N-1},U_{N-1}}\bigg)
\end{align*}

For \(V\) to be optimal, we must have
\[
C(V) \leq C(V') 
\]
\[
\iff \alpha^{k - j}\, I^\ast_i + \alpha^{(k+1) - j'}\, I^\ast_{i'} 
    \;\leq\; \alpha^{k - j'}\, I^\ast_{i'} + \alpha^{(k+1) - j}\, I^\ast_i
\]
\[
\iff \frac{I^\ast_i}{\alpha^{j - k}} \;\leq\; \frac{I^\ast_{i'}}{\alpha^{j' - k}},
\]
where \(\alpha^k\) cancels on both sides of the inequality. This shows that a later observation (larger \(j\)) is better. A useful way to interpret the discount factor $\alpha$ is as a geometrically distributed ``memory window" of effective size $\frac{1}{1-\alpha}$. Thus, a smaller $\alpha$ (e.g., $\alpha=0.2$) corresponds to a shorter effective window, which aggressively favors fresh, high-information samples, whereas a larger $\alpha$ (e.g., $\alpha=0.8$) corresponds to placing more weight on older samples as well. 

Therefore, if we want to transmit the {\em most valuable} sample first, we should pick the sample whose ratio \(\frac{I^\ast_i}{\alpha^{j - k}}\) is the largest among the queued samples. In other words, we schedule the samples in the descending order of \(\frac{I^\ast_i}{\alpha^{j - k}}\). 

\subsection{The Look-Back Window Scheduling Policy}\label{subsec:look_back_win_pol}

Another heuristic approach for scheduling non-homogeneous sensor observations is the look-back window policy. This policy is a variation of the usual LCFS  rule by incorporating information quality into the scheduling decision. Recall from Section~\ref{sec:non-fcfs-discipline} that LCFS is conjectured to minimize detection delay in the single-sensor case. In the multi-sensor setting, scheduling must also account for the trade-off between sensor quality and measurement packet recency. This motivates policies like the look-back window, which generalizes LCFS to balance freshness and informativeness.

Instead of always transmitting the most recent sample, the system considers the most informative sample among the most recent $w$ arrivals. At each transmission opportunity, if the most recent sample is from a sensor with     higher information content, it is transmitted immediately. However, if the most recent sample comes from a sensor with  lower information content, the system searches within the last $w$ arrivals for a more informative sample—that is, a sample from a sensor with higher information content. If such a sample exists, it is prioritized for transmission; otherwise, the most recent sample is transmitted. Note that if $w=1$, the policy reduces to LCFS. This policy balances freshness and information content, ensuring that older but more informative samples are not ignored while avoiding excessive delays. 

\section{Discussion and numerical analysis}
\label{sec:practical-considerations-discussion}

In \cref{sec:infin-retr-in-order-reception}, we analyzed the asymptotic performance of the \CUSUM{} algorithm as $\text{ARL2FA}\!\to\!\infty$. Moustakides et al.~\cite{moustakidesMarkovComments2010} observed that for Markov observations, the optimal threshold in the non-asymptotic regime depends on the initial state. Using Lai’s asymptotic optimality framework~\cite{laiInformationBoundsQuick1998}, we show in \cref{sec:asymptotic-analysis-geom-geom-one-model} that our \CUSUM{} algorithm (\cref{eq:cusum-definition}) achieves asymptotic performance independent of the initial state. As $\text{ARL2FA}\!\to\!\infty$, the threshold $h\!\to\!\infty$, and the Markov process $\{\zeta\}$ approaches its stationary post-change distribution, rendering the initial state irrelevant.

\Cref{fig:add-close-to-hbyI} plots the ratio $\text{ADD}/h$ versus the average sampling interval $s=1/r$ for the single-sensor case with unbounded retransmissions. The system follows the model in \cref{sec:system-model} with Bernoulli sampling rate $r$, packet success probabilities $(p_0,p_1)=(0.61,0.60)$, and sensor distributions $f_0=\mathcal{N}(0,1/2)$, $f_1=\mathcal{N}(10,1/2)$. The initial queue length $Q_1$ is drawn from the stationary distribution with arrival rate $r$ and service rate $p_0$. For each run, we fix $\nu=1$ and simulate until the \CUSUM{} statistic crosses $h$, averaging $\text{ADD}$ over $10^6$ repetitions. As shown, $\text{ADD}/h$ approaches $1/I$ (see \cref{thm:upper-bound-cusum}) as $h\!\to\!\infty$, validating our asymptotic analysis.

\begin{figure}[h]
\centering
\includegraphics[width=\columnwidth, height=6cm]{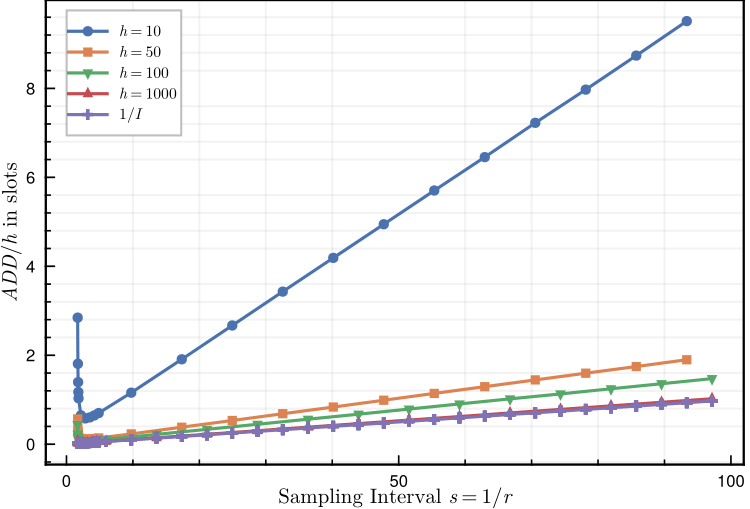}
\caption{Single sensor case: QCD over a lossy link with unbounded retransmissions; \CUSUM{} algorithm in \cref{eq:cusum-definition}: Simulated `$\text{ADD}/h$' vs Average Sampling Interval}
\label{fig:add-close-to-hbyI}
\end{figure}

\begin{figure}[h]
\centering
\includegraphics[clip, trim=0cm 0cm 0cm .62cm, width=\columnwidth,height=6cm]{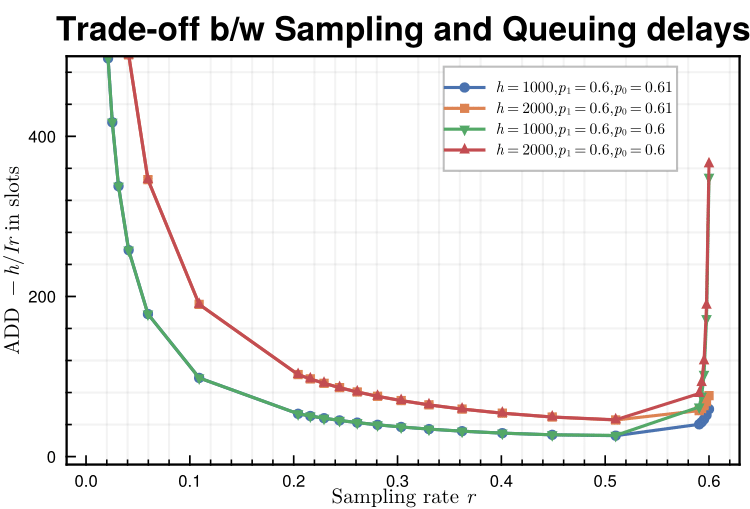}
\caption{Single sensor case: QCD over a lossy link with unbounded retransmissions; \CUSUM{} algorithm in \cref{eq:cusum-definition}: Trade-off between sampling delay and queuing delay}.
\label{fig:add_vs_s_constant_h_smoothsim_tick_curve}
\end{figure}

To ensure queue stability, we maintain $r < \min\{p_0,p_1\}$. In practice, the transmit queue adds a delay $\overline{d}_Q$ comprising (a) the time to clear packets present at slot~1 and (b) delay experienced by the packet containing the decisive measurement. As $r\!\to\!p_0$, the former dominates, while near $p_1$ the latter increases. Hence, smaller $r$ reduces $\overline{d}_Q$. However, the average sampling delay $\overline{d}_S=1/r$ increases with smaller $r$. \Cref{fig:add_vs_s_constant_h_smoothsim_tick_curve} illustrates this trade-off: increasing $r$ reduces $\overline{d}_S$ but inflates $\overline{d}_Q$ near $p_1$. We plot $\text{ADD}$ for QCD in the single sensor case using the \CUSUM{} algorithm from \cref{eq:cusum-definition}, sweeping $r$ across two threshold values $h$. Simulation parameters match those used in \cref{fig:add-close-to-hbyI}. As $r$ nears $p_1=0.6$, $\text{ADD}$ spikes due to increased $\overline{d}_Q$. Notice that the curves in \cref{fig:add_vs_s_constant_h_smoothsim_tick_curve} with different $p_0$ differ only when $r$ approaches $p_1$, reflecting the time to empty the initial transmit buffer. The plot also shows that for given the pair $\left( p_0, p_1 \right)$, an optimal sampling rate exists that achieves the lowest detection penalty due to sampling and queuing delays. A network-aware DM can exploit this trade-off to choose an optimal $r$, whereas a network-oblivious DM, lacking knowledge of channel statistics, may incur larger delay penalties due to a mismatched sampling rate.

Recall that the asymptotic ESADD for the QCD procedure with unbounded retransmissions is $\tfrac{h}{I}(1 + \smallO{1})$. In \cref{fig:add-close-to-hbyI}, the gap between the theoretical curve $\text{ESADD}/h = 1/I$ and simulated $\text{ADD}/h$ corresponds to the $\smallO{1}$ term, which shrinks as $h$ increases. In practice, this deviation arises from queuing and sampling delays, i.e., $(\overline{d}_S + \overline{d}_Q)/h$. Since both delays grow sublinearly with $h$, the simulated $\text{ADD}/h$ approaches $1/I$ for large thresholds.

\begin{figure}[bt] \centering
\includegraphics[clip, trim=0cm 0cm 0cm .55cm, width=\columnwidth,height=6cm]{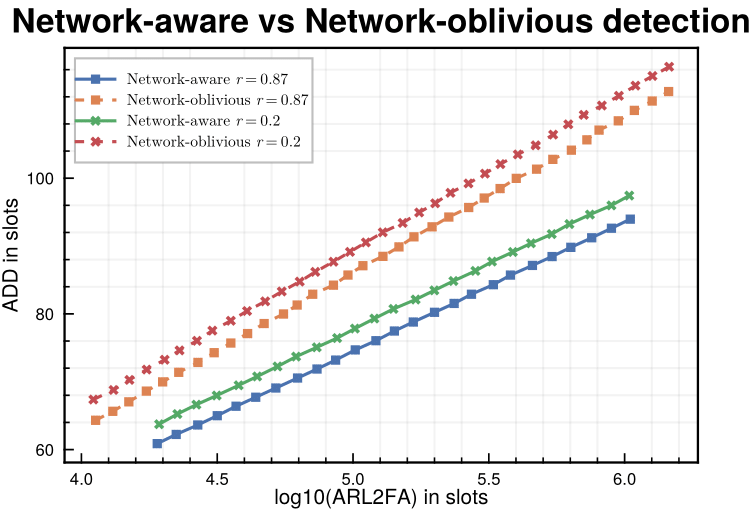}
\caption{Single sensor case: QCD over a lossy link with unbounded retransmissions; \CUSUM{} algorithm in \cref{eq:cusum-definition}: Network-aware vs Network-oblivious detection}
\label{fig:add_vs_arl2fa_compare_NA_NO}
\end{figure}

\Cref{fig:add_vs_arl2fa_compare_NA_NO} plots simulated $\text{ADD}$ versus $\text{ARL2FA}$ for different sampling rates $r$. Parameters are $p_0=0.95$, $p_1=0.90$, $f_0=\mathcal{N}(0,1/2)$, and $f_1=\mathcal{N}(1,1/2)$. The plot confirms the well-known linear growth of $\text{ADD}$ with $\log(\text{ARL2FA})$~\cite{tartakovskySequentialAnalysisHypothesis2015}.  It also highlights the consistent advantage of the proposed network-aware \CUSUM{} detector, which utilizes the channel process $\{Y_k\}$, over a network-oblivious detector that updates only upon successful receptions.

\begin{figure}[h]
\centering
    \includegraphics[clip, trim=0cm 0cm 0cm 0cm, width=\columnwidth, height=6cm]{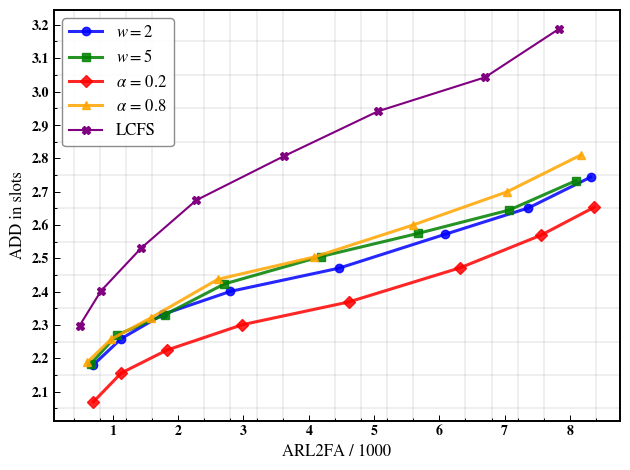}
    \caption{Multiple non-homogeneous sensors: Simulated $\text{ADD}$ vs ARL2FA with $N=5$ sensors. The sampling probabilities are chosen to be $r_i \in \{0.13,0.15,0.2,0.22,0.24\}$ and $p=0.91$, and each sensor's measurement standard deviation is $\sigma_i \in \{1.65, 2, 0.7, 1.75, 1.5\}$. Pre- and post-change means are set to 0 and 5, respectively. The change occurs at time slot 100, and each simulation runs for 10{,}000 slots over 500 sample paths.}
\label{fig:MS-del vs arl r-0.2}
\end{figure}
\begin{figure}[bt]
    \begin{centering}
    \begin{center}
        \includegraphics[clip, trim=0cm 0cm 0cm 0cm, width=\columnwidth, height=6cm]{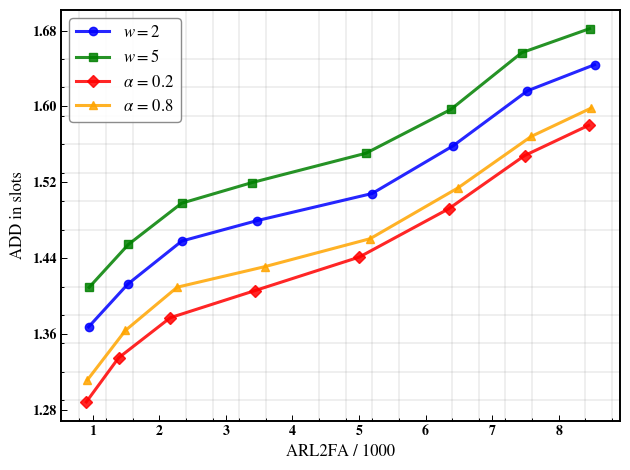}
    \end{center}
    \end{centering}
    \vspace{-5mm}
    \caption{Multiple non-homogeneous sensors: Simulated $\text{ADD}$ vs
  ARL2FA with same parameters as \cref{fig:MS-del vs arl r-0.2} except sampling probabilities are $r_i \in \{0.65, 0.60, 0.47, 0.65, 0.69\}$ and $p=0.95$.}
    \vspace{-1mm}
\label{fig:MS-del vs arl r-0.65}
\end{figure}
In \cref{fig:MS-del vs arl r-0.2}, we compare the performance of different scheduling policies under a non-homogeneous sensor setup where the channel loss is independent of the change. The figure presents results for the look-back scheduling policy with two look-back window sizes $(w=2, 5)$, a discounted information scheduler with discount factors $\alpha=0.2$ and $\alpha=0.8$, and a Last-Come-First-Served (LCFS) benchmark.

The look-back scheduling policy with $w=2$ shows a substantial improvement over $w=1$ (i.e., LCFS), but going from $w=2$ to $w=5$ yields only marginal additional gains. This suggests that while a small look-back window is helpful in selecting informative measurements slightly older than the very latest sample, increasing the window size further saturates the improvement especially in this setup, where limited sampling rates mean relatively few samples are available within the window at any time. Similarly for the discounted information scheduler, we observe that $\alpha=0.2$ leads to the lowest detection delays because it strongly emphasizes the most recent, high-information observations, akin to a very short effective memory. Increasing $\alpha$ to $0.8$ introduces a larger memory window, which slightly raises detection delay because older samples receive relatively more weight, potentially causing the scheduler to transmit measurements that, while still informative, are not as fresh. Importantly, both heuristics—the look-back scheduler and the discounted information scheduler can be tailored to outperform LCFS, as clearly demonstrated in \cref{fig:MS-del vs arl r-0.2}. By tuning parameters such as $w$ and $\alpha$, these policies adaptively balance freshness and informativeness, leading to lower detection delays than LCFS.

In \cref{fig:MS-del vs arl r-0.65}, we analyze the same scheduling policies under an increased sampling probability scenario and channel success probability is set to $p=0.95$. The overall trend remains consistent, with the look-back scheduling policies and discounted information scheduler achieving comparable performance. However, detection delays decrease for all policies due to the higher sampling probabilities, which increase the frequency of new samples and reduce the time needed to accumulate sufficient LLR values for detection.

\section{Conclusion}
\label{sec:conclusion-networked-qcd}
We have studied the quickest change detection (QCD) problem in a networked setting where wireless channel losses introduce delays and out-of-order reception of measurements at the decision maker. Under a single-sensor model with a channel loss probability that changes after the change point, we proposed a \CUSUM{} algorithm, established its asymptotic optimality as the false alarm rate tends to zero by modeling the observation process as Markovian, and extended the analysis to the case of independent measurement losses. We further examined the effect of transmit queue disciplines and showed that, under likelihood ratio ordering of the pre- and post-change distributions, the CUSUM statistic for the FCFS, any other non-idling discipline, and the LCFS are stochastically ascendingly ordered. Extending to the multi-sensor setting, we analyzed homogeneous and heterogeneous sensor networks, highlighting the role of scheduling in balancing timeliness and informativeness. We proposed heuristic scheduling policies—such as discounted-information and look-back window rules—and demonstrated their impact on non-asymptotic detection delay. 

Overall, this work bridges sequential detection theory and wireless communication dynamics by incorporating queueing and retransmission effects into the QCD framework. Future directions include developing provably optimal scheduling policies, extending the analysis to correlated sensors and unknown post-change distributions of the observations.

\bibliographystyle{IEEEtran}
\bibliography{References.bib}

\appendices 
\section{Proof of \cref{lem:log-likelihood-computation}}\label{appendix:Proof of lem:log-likelihood-computation}
We write
\begin{eqnarray}\label{eq:log-likelihood-lem1-first-step}
    \ell_n &=& \sum_{i=1}^n \log
\frac{\mathbb{P}_{i, \nu}\left( Y_i \mid Q_1, Y_1^{i-1} \right)}{\mathbb{P}_{i,
\infty}\left( Y_i \mid Q_1, Y_1^{i-1} \right)} \nonumber\\
&\quad+&  \sum_{i=1}^n \log
\frac{\mathbb{P}_{i, \nu}\left( Z_i \mid Q_1, Y_1^i, Z_1^{i-1} \right)}{\mathbb{P}_{i,
\infty}\left( Z_i \mid Q_1, Y_1^i, Z_1^{i-1} \right)},
\end{eqnarray}
where $\mathbb{P}_{i, \nu}$ and $\mathbb{P}_{i, \infty}$ are the
probability distributions at time $i$, given that the change occurs at  a finite $\nu$ and  
$\nu = \infty$ (no change), respectively.

In the first summand, we drop the
dependence of $Y_i$ on $Z_1^i$ since $Y_i\independent Z_i$ given $\left(Q_1,Y_1^{i-1}\right)$.  Next, for $1 \leq i \leq n$, we write
\small
\begin{align*}
  P\left( Y_i \mid Q_1, Y_1^{i-1} \right) &= P\left( Y_i \mid Q_i > 0 \right) P\left( Q_i > 0 \mid Q_1, Y_1^{i-1} \right) \\
  &\quad+ P\left( Y_i \mid Q_i = 0 \right) P\left( Q_i = 0 \mid Q_1, Y_1^{i-1} \right).
\end{align*}
\normalsize
Next, we make the following observations:
\begin{itemize}
\item Under both the probability laws $\mathbb{P}_\nu$ and $\mathbb{P}_\infty$, for $1\leq i \leq n$, we have
  \begin{displaymath}
    P\left( Y_i = 1 \mid Q_i = 0 \right) = P\left( Y_i = 0 \mid Q_i = 0 \right) = 0,
  \end{displaymath}
  since $Y_i = \emptyset$ w.p.~1 whenever the transmit queue is empty ($Q_i = 0$).
\item For each $1\leq i \leq n$, the probability
  $P\left( Q_i \mid Q_1, Y_1^{i} \right)$ is equal under both $\mathbb{P}_\nu$
  and $\mathbb{P}_\infty$ because the sampling process (by which the arrivals
  enter the transmitter queue) is independent of the change point
  $\nu$.
\end{itemize}
Therefore, for $1 \leq i \leq n$, we have:
\begin{align*} \log{\frac{\mathbb{P}_{i, \nu}\left( Y_i \mid Q_1, Y_1^{i-1} \right)}{\mathbb{P}_{i, \infty}\left( Y_i \mid Q_1, Y_1^{i-1} \right)}} =
  \begin{cases}
    \log{\frac{p_1}{p_0}}, &\text{if }Y_i = 1,\\
    \log{\frac{1-p_1}{1-p_0}}, &\text{if }Y_i = 0,\\
    0, &\text{if }Y_i = \emptyset.
  \end{cases}
\end{align*}
That is, $\log{\frac{\mathbb{P}_{i, \nu}\left( Y_i \mid Q_1, Y_1^{i-1} \right)}{\mathbb{P}_{i, \infty}\left( Y_i \mid Q_1, Y_1^{i-1} \right)}} = \indicator{ Q_i > 0
}\log\frac{\mathbb{P}_0\left( Y_i \mid Q_i > 0 \right)}{\mathbb{P}_\infty\left( Y_i
    \mid Q_i > 0 \right)}$, so that we can simplify the first summand in eq. (\ref{eq:log-likelihood-lem1-first-step}) to:
    \small
\begin{align*}
  \sum_{i=1}^n \log
  \frac{\mathbb{P}_{i, \nu}\left( Y_i \mid Q_1, Y_1^{i-1} \right)}{\mathbb{P}_{i,
      \infty}\left( Y_i \mid Q_1, Y_1^{i-1} \right)} = \sum_{i = \nu + 1}^n \indicator{ Q_i > 0
  }\log\frac{\mathbb{P}_0\left( Y_i \mid Q_i > 0 \right)}{\mathbb{P}_\infty\left( Y_i
      \mid Q_i > 0 \right)}.
\end{align*}
\normalsize
To simplify the second term in eq. (\ref{eq:log-likelihood-lem1-first-step}), note that $D_i$ can be determined using $\left( Q_1, Y_1^{i}\right)$, and that $Z_i$ depends only on $\left(Y_i, D_i\right)$. Hence, \(P\left( Z_i \mid Q_1, Y_1^{i}, Z_1^{i-1} \right) = P\left( Z_i \mid D_i, Y_i\right)\) under both the probability laws. Further, on the set where \(\left\{ Y_i = 0 \right\}\), we have \(Z_i = \emptyset\) with probability $1$ under both the probability laws. Also, on the sample paths where \(\left\{ Y_i = 1 \right\}\), we have \(Z_i = X_{D_i}\). Hence, we write
\small
\begin{align*}
 \log \frac{\mathbb{P}_{i,\nu}\left( Z_i \mid Q_1, Y_1^{i}, Z_1^{i-1} \right)}{\mathbb{P}_{i,\infty}\left( Z_i \mid Q_1, Y_1^{i}, Z_1^{i-1} \right)} &= \indicator{ Y_i = 1 }\log \frac{\mathbb{P}_{i, \nu}\left(Z_i \mid D_i, Y_i = 1 \right)}{\mathbb{P}_{i, \infty}\left( Z_i \mid D_i, Y_i=1 \right)} \\
 &= \indicator{Y_i = 1, D_i \geq s_\nu + Q_1}\log\frac{f_1\left( X_{D_i} \right)}{f_0\left( X_{D_i} \right)}.
\end{align*}
\normalsize
Thus, the second term in eq. (\ref{eq:log-likelihood-lem1-first-step}) can be written as
\begin{align*}
&\sum_{i=1}^n \log \frac{\mathbb{P}_{i, \nu}\left( Z_i \mid Q_1, Y_1^i, Z_1^{i-1} \right)}{\mathbb{P}_{i,\infty}\left( Z_i \mid Q_1, Y_1^i, Z_1^{i-1} \right)}\\
&= \sum_{i=\nu+1}^{n} \indicator{ Y_i = 1, D_i \geq s_{\nu} + Q_1 } \log \frac{f_1\left( X_{D_i} \right)}{f_0\left( X_{D_i} \right)}.
\end{align*}

\section{Proof of \cref{thm:lower-bound-cusum}}\label{appendix:Proof of thm:lower-bound-cusum}
We will use the following lemma to prove the lower bound on the
$\text{ESADD}$ (\cref{thm:lower-bound-cusum}).
\begin{lemma}
  \label{lem:asconvergence-implies-max-convergence}
Given $\left\{R_k: k\geq1\right\}$, a sequence of i.i.d.\ random variables with $\expect{R_1} < \infty$, if $\frac{1}{n}
\sum_{i=1}^{n}R_i \xrightarrow{a.s.} 0$, then
\begin{displaymath}
\lim_{n\to\infty}P\bigg( \max_{k \leq
n}\frac{1}{n} \sum_{i=1}^{k} R_i \geq \delta \bigg) = 0, \forall
\delta > 0.
 \end{displaymath}
\end{lemma}

\begin{proof}
 Define the random variable $K_n = \argmax_{k\leq n} \frac{1}{n} \sum_{i=1}^k
 R_i$. For each $\omega$ in the set
\(\left\{ \omega
: \max_{k \leq n}\frac{1}{n} \sum_{i=1}^{k}R_i \geq \delta \right\},\)
we have  $K_n(\omega) \leq n$, and $K_n(\omega) \to \infty$ as
$n\to\infty$. Thus, $\forall \delta > 0$,
\begin{displaymath}
  P\left( \frac{1}{n} \sum_{i=1}^{K_n} R_i \geq
\delta \right) \leq P\left( \frac{1}{K_n}\sum_{i=1}^{K_n}R_i \geq
\delta \right).
\end{displaymath}
The quantity on the right hand side of the above equation $P\left( \frac{1}{K_n}\sum_{i=1}^{K_n}R_i \geq
\delta \right)$ goes to zero as $n \to \infty$
since $\frac{1}{n}\sum_{i=1}^n R_i \xrightarrow[]{a.s.} 0$ implies in probability convergence. This concludes the proof.
\end{proof}

Now to prove \cref{thm:lower-bound-cusum}, we make use of Lai's \cite[Theorem 1]{laiInformationBoundsQuick1998} lower bound on the asymptotic $\text{ESADD}$. Following the discussion in \cite[Sec.~IV]{laiInformationBoundsQuick1998}, we only need to show
\begin{equation} \label{eq:lower-bound-lai-sufficient-condition}
    \lim_{n\to\infty} \sup_{x} \mathbb{P}_1\left\{
      \max_{t\leq n} \sum_{i=1}^{t} L_i \geq I\left( 1 + \delta \right)n
      \mid \zeta_0 = x\right\} = 0,
  \end{equation}
for each $\delta > 0$. From \cref{eq:ergodic-limit-Li}, we obtain $\frac{1}{n} \sum_{i=1}^{n}L_i \xrightarrow{a.s.} I$. For a fixed $x$, set $\lambda_i = L_i - I$, and apply \cref{lem:asconvergence-implies-max-convergence} to show that 
\begin{displaymath}
   \forall x, \forall \delta > 0, \mathbb{P}_1\left\{ \max_{t\leq n} \frac{1}{n} \sum_{i=1}^{t}\lambda_i >
    \delta \mid \zeta_0 = x \right\} \to 0.
\end{displaymath}
To show that \cref{eq:lower-bound-lai-sufficient-condition} is true, it is sufficient to show that \(\sup_{\zeta_0 = x} \frac{1}{n} \sum_{i=1}^n L_i \xrightarrow{a.s.} I\). We note that since $\left\{ \zeta_k \right\}$ is a Markov process, it is sufficient  to only show that \(\sup_{\zeta_0 = x} L_1 < \infty\). This can be easily observed by noting that the first term in $L_1$ (see \cref{eq:define-Li}) is bounded whenever \(0 < p_0,p_1 < \infty\), and the second term in $L_1$ depends on $\zeta_0$ only through indicator functions. This proves the assertion that \(\sup_{\zeta_0 = x}\frac{1}{n} \sum_{i=1}^{n}L_i \xrightarrow{a.s.} I\), and hence the theorem.

\section{Proof of \cref{thm:upper-bound-cusum}}\label{appendix:Proof of thm:upper-bound-cusum}
To prove this claim, we make use of Lai's \cite[Theorem. 4]{laiInformationBoundsQuick1998} upper bound on the asymptotic ESADD, for a \CUSUM{} detector, with threshold $h$. As before, following the discussion in \cite[Sec.~IV]{laiInformationBoundsQuick1998}, we need to show that
\begin{equation}\label{eq:upper-bound-lai-sufficient-condition}
    \lim_{n\to\infty} \sup_{x} \mathbb{P}_1\left\{ \sum_{i=1}^{n}L_i < \left( I-\delta \right) n \mid \zeta_0 = x \right\} =0,
\end{equation}
For each $x$, the limit
\begin{displaymath}
      \forall \delta > 0, \mathbb{P}_1\left\{  \frac{1}{n} \sum_{i=1}^{n}L_i >
        \delta \mid \zeta_0 = x \right\} \to 0
    \end{displaymath}
    is true since we have from \cref{eq:ergodic-limit-Li} that
    $\frac{1}{n}\sum_{i=1}^n L_i \xrightarrow{a.s.} I$ and almost sure convergence
    implies in probability convergence of $\frac{1}{n}\sum_{i=1}^n L_i$ to $I$.   
  To prove 
  \cref{eq:upper-bound-lai-sufficient-condition}, we follow 
  a similar approach as in the proof of \cref{thm:lower-bound-cusum}.

\section{Proof of \cref{lem:non-fifo-log-likelihood-ratio}}\label{appendix:Proof of lem:non-fifo-log-likelihood-ratio}
We write $\ell_{\nu, n}$ as
  \begin{displaymath}
    \ell_{\nu, n} = \log\frac{\mathbb{P}_\nu\left( Y_1^n \mid Q_1 \right)}{\mathbb{P}_\infty\left( Y_1^n \mid Q_1 \right)} + \log\frac{\mathbb{P}_\nu\left( J_1^n, Z_1^n \mid Q_1, Y_1^n \right)}{\mathbb{P}_\infty\left( J_1^n, Z_1^n \mid Q_1, Y_1^n \right)}.
  \end{displaymath}
Note that the second term in the above equation can be written as
\begin{align*}
    &\log\frac{\mathbb{P}_\nu\left( J_1^n, Z_1^n \mid Q_1, Y_1^n \right)}{\mathbb{P}_\infty\left( J_1^n, Z_1^n \mid Q_1, Y_1^n \right)} \\
    &= \sum_{k=1}^n\log\frac{\mathbb{P}_\nu\left( J_k, Z_k \mid Q_1, Y_1^k, Z_1^{k-1}, J_1^{k-1} \right)}{\mathbb{P}_\infty\left( J_k, Z_k \mid Q_1, Y_1^k, Z_1^{k-1}, J_1^{k-1} \right)},
\end{align*}
where we use the fact that \(J_k, Z_k \independent Y_{k+1}^n\). Next, simplify the conditional joint probability \(P\left( J_k, Z_k \mid Q_1, Y_1^k, Z_1^{k-1}, J_1^{k-1} \right) \), under both the probability laws, as
\begin{align*}
    &P\left( J_k, Z_k \mid Q_1, Y_1^k, Z_1^{k-1}, J_1^{k-1} \right) \\
    &= P\left( J_k \mid Y_1^k, Z_1^{k-1}, J_1^{k-1} \right) P\left( Z_k \mid Y_1^k, J_1^k, Z_1^{k-1} \right) \\
    &= P\left( J_k \mid Y_1^k, J_1^{k-1} \right)P\left( Z_k \mid Y_1^k, J_1^k, Z_1^{k-1} \right),
\end{align*}
where we note that $J_k \independent Z_1^{k-1}$ given $Y_1^k$. Further, the conditional probability $P\left( J_k \mid Y_1^k, J_1^{k-1} \right)$ is the same for $k = 1, \ldots, n$ under both the probability laws $\mathbb{P}_\nu$ and $\mathbb{P}_\infty$. Hence, we write
\begin{displaymath}
    \log\frac{\mathbb{P}_\nu\left( J_1^n, Z_1^n \mid Q_1, Y_1^n \right)}{\mathbb{P}_\infty\left( J_1^n, Z_1^n \mid Q_1, Y_1^n \right)} =
    \log\frac{\mathbb{P}_\nu\left( Z_1^n \mid Q_1, Y_1^n, J_1^n \right)}{\mathbb{P}_\infty\left( Z_1^n \mid Q_1, Y_1^n, J_1^n \right)}.
\end{displaymath}
The right hand side of the above equation simplifies to the desired form following the steps in the second part of the proof of Lemma~\ref{lem:log-likelihood-computation}.

\section{Proof of \cref{prop:non-fcfs-cusums-lr-ordered}}\label{appendix:Proof of prop:non-fcfs-cusums-lr-ordered}

Under the hypothesis that $f_0 \lrleq f_1$, the log-likelihood ratio $\log\frac{f_1\left(\cdot\right)}{f_0\left(\cdot\right)}$ is a monotone function.  The proof uses Lemma~7 and the property that monotone functions of LR ordered random variables are LR ordered \cite[Thm.~1.C.8]{shakedStochasticOrders2007}. We will first prove the first inequality that $C_m^{\left( F \right)} \leq_{st} C_m^{\left( A \right)}$. 

From \cite[Thm.~1.C.8]{shakedStochasticOrders2007}, we have, for each $j, n \in \mathbb{N}, j\leq n$,
\begin{displaymath}
  \log\frac{f_1\left(X_j\right)}{f_0\left(X_j\right)} \lrleq \log\frac{f_1\left(X_{k_n^{\left(j\right)}}\right)}{f_0\left(X_{k_n^{\left(j\right)}}\right)},
\end{displaymath}
where $\left\{k_n^{\left(j\right)} : j\leq n\right\}$ are the sampling indices of the received measurements at the DM after $n$ samples that are rearranged in the ascending order. This implies ordering in the stochastic order sense. Given pairs of stochastically ordered random variables, here $(j, k^{(j)}_n)$, $1 \leq j \leq n$, using the property that the stochastic order is closed under convolutions \cite[Thm.~1.A.3.(b)]{shakedStochasticOrders2007}, we have, for each $i \leq n$,
\begin{displaymath}
  \sum_{j = i}^n \log\frac{f_1\left(X_j\right)}{f_0\left(X_j\right)} \leq_{st} \sum_{j = i}^n \log\frac{f_1\left(X_{k_n^{\left(j\right)}}\right)}{f_0\left(X_{k_n^{\left(j\right)}}\right)}.
\end{displaymath}
Thus, for each slot $m\in\N$, we have
\begin{displaymath}
  \sum_{j = i}^{D_m} \log\frac{f_1\left(X_j\right)}{f_0\left(X_j\right)} \leq_{st} \sum_{j = i}^{D_m} \log\frac{f_1\left(X_{k_{D_m}^{\left(j\right)}}\right)}{f_0\left(X_{k_{D_m}^{\left(j\right)}}\right)}.
\end{displaymath}
Suppose that \(s_i = \min\left\{ m' \in\N: D_{m'} = i \right\}\), then for each sample \(i \leq n\) and each slot index \(s \in \left\{ s_i, s_i+1, \ldots, s_{i+1}-1 \right\}\) between the $i$th and $i+1$\textsuperscript{st} samples, we have
\begin{align*}
  \sum_{b = s}^m\indicator{Q_b > 0}\log\frac{\mathbb{P}_s\left( Y_b \mid Q_b > 0 \right)}{\mathbb{P}_\infty\left( Y_b \mid Q_b > 0 \right)} + \sum_{j = i}^{D_m} \log\frac{f_1\left(X_j\right)}{f_0\left(X_j\right)} \leq_{st} \\  
  \sum_{b = s}^m\indicator{Q_b > 0}\log\frac{\mathbb{P}_s\left( Y_b \mid Q_b > 0 \right)}{\mathbb{P}_\infty\left( Y_b \mid Q_b > 0 \right)} + \sum_{j = i}^{D_m} \log\frac{f_1\left(X_{k_{D_m}^{\left(j\right)}}\right)}{f_0\left(X_{k_{D_m}^{\left(j\right)}}\right)}.
\end{align*}
Note that
\small
\begin{align*}
  \ell^{\left( F \right)}_{s, m} &= \sum_{b = s}^m\indicator{Q_b > 0}\log\frac{\mathbb{P}_s\left( Y_b \mid Q_b > 0 \right)}{\mathbb{P}_\infty\left( Y_b \mid Q_b > 0 \right)} + \sum_{j = i}^{D_m} \log\frac{f_1\left(X_j\right)}{f_0\left(X_j\right)}, \\
  \ell^{\left( A,m \right)}_{s, m} &= \sum_{b = s}^n\indicator{Q_b > 0}\log\frac{\mathbb{P}_s\left( Y_b \mid Q_b > 0 \right)}{\mathbb{P}_\infty\left( Y_b \mid Q_b > 0 \right)} + \sum_{j = i}^{D_m} \log\frac{f_1\left(X_{k_{D_m}^{\left(j\right)}}\right)}{f_0\left(X_{k_{D_m}^{\left(j\right)}}\right)},
\end{align*}
\normalsize
where $\ell^{\left( F \right)}_{s, m}$ and $\ell^{\left( A,m \right)}_{s, m}$ denote the log-likelihood ratio computed at the end of the $m$\textsuperscript{th} slot given that the change point is $s$, when the FCFS queuing discipline and an alternate queuing discipline are used in the transmit queue respectively. That is, we have $\ell^{\left( F \right)}_{s, m} \leq_{st} \ell^{\left( A,m \right)}_{s, m}$ for each $s \leq m$ and for each $m\in\N$. $C_m^{(F)} = \max_{1 \leq s \leq m} l_{s,m}^{(F)}$, and similarly for $C_m^{(A)}$. However, the max is made of dependent random variables. To show $C_m^{(F)} \leq_{st} C_m^{(A)}$, we now argue via induction.

Fix $m$. Observe that $l^{(F)}_{s,m} = \sum_{b=s}^m l^{(F)}_b$, where 
$$
l^{(F)}_b = \indicator{Q_b > 0}\log\frac{\mathbb{P}_s\left( Y_b \mid Q_b > 0 \right)}{\mathbb{P}_\infty\left( Y_b \mid Q_b > 0 \right)} + \indicator{b=s_j}  \log\frac{f_1\left(X_j\right)}{f_0\left(X_j\right)}.
$$
Similarly, $l^{(A,m)}_{s,m} = \sum_{b=s}^m l^{(A,m)}_b$, where 
$$
l^{(A,m)}_b = \indicator{Q_b > 0}\log\frac{\mathbb{P}_s\left( Y_b \mid Q_b > 0 \right)}{\mathbb{P}_\infty\left( Y_b \mid Q_b > 0 \right)} + \indicator{b=s_j}  \log\frac{f_1\left(X_{k^{(j)}_{D_m}}\right)}{f_0\left(X_{k^{(j)}_{D_m}}\right)}.
$$
We then have $l_b^{(F)} \leq_{st} l_b^{(A,m)}$, for each $1 \leq b \leq m$, because convolution preserves stochastic orders.

Consider the CUSUM $\tilde{C}_{b}^{(A,m)} = \max\{ \tilde{C}_{b-1}^{(A,m)} + l_b^{(A,m)}, 0\}$, $1 \leq b \leq m$, with initialization $\tilde{C}_0^{(A,m)} = 0$; then $C^{(A)}_m = \tilde{C}_m^{(A,m)}$, i.e., the statistic is the CUSUM on the reordered samples. Observe that on alternative service disciplines, the reordering and the entire CUSUM chain of computation must be done every time a sample is received. For the FCFS service discipline, of course, the usual CUSUM applies, and  $C_b^{(F)} = \max\{C_{b-1}^{(F)} + l_b^{(F)}, 0\}$, $1 \leq b \leq m$,  with $C_0^{(F)} = 0$.

We now induct on $b$ from $0$ to $m$. Clearly, $C_b^{(F)} \leq_{st} \tilde{C}_{b}^{(A,m)}$ for $b=0$ because both are identically 0. Assume its validity for $b<m$. Now $l_{b+1}^{(F)} \leq_{st} l_{b+1}^{(A,m)}$. Further, $l_{b+1}^{(F)}$ is independent of $C_b^{(F)}$ and similarly $l_{b+1}^{(A,m)}$ is independent of $\tilde{C}_b^{(A,m)}$. Stochastic order is closed under convolution \cite[Thm.~1.A.3.(b)]{shakedStochasticOrders2007}, and so $C_b^{(F)} + l_{b+1}^{(F)} \leq_{st} \tilde{C}_b^{(A,m)} + l_{b+1}^{(A,m)}$. Stochastic order is closed under the monotone operation $\max\{\cdot,0\}$, and so $C_{b+1}^{(F)} \leq_{st} \tilde{C}_{b+1}^{(A,m)}$. This establishes that $C_m^{(F)} \leq_{st} \tilde{C}_m^{(A,m)} = C_m^{(A)}$.

The proof for the second inequality is similar. Hence, we have for each $m\in\N$, $C_m^{\left( F \right)} \leq_{st} C_m^{\left( A \right)} \leq_{st} C_m^{\left( L \right)}.$

\section{Discussion related to \cref{thm:cusums-are-st-ordered}}\label{appendix:disc of thm:cusums-are-st-ordered}

Recall the discussion in
Section~IV.A that the threshold $h$ for a target $\text{ARL2FA} = \,\gamma$ is the same for any transmit queue discipline used at the transmitter. Fix this $h$ at $h_{\gamma}$.

Observe that $\{N^{(F)}(h) > m\} = \{C_b^{(F)} \leq h, b = 1, \ldots, m\}.$ Since $C_b^{(F)} = \max_{1 \leq s \leq b} l_{s,b}^{(F)}$, where $l_{s,b}^{(F)}$ is the partial sum of the likelihoods from $s$ to $b$, we observe that
\begin{align}
  \{N^{(F)}(h) > m\} 
  & = \{C_b^{(F)} \leq h, b = 1, \ldots, m\} \nonumber \\
  & = \{ l_{s,b}^{(F)} \leq h, 1 \leq s \leq b \leq m\}.
\end{align}
Similarly, for an alternative service discipline, we get
\begin{align}
  \{N^{(A)}(h) > m\} = \{ l_{s,b}^{(A,b)} \leq h, 1 \leq s \leq b \leq m\}.
\end{align}
The above set equivalences and $l_{s,b}^{(F)} \leq_{LR} l_{s,b}^{(A,b)}$ for every $1 \leq s \leq b \leq m$ suggests that $\mathbb{P}(N^{(A)}(h) > m) \leq \mathbb{P}(N^{(F)}(h) > m)$ may hold, but a proof eludes us because of the dependencies in the random variables $\{ l_{s,b}^{(A,m)} \}_{1 \leq s \leq b \leq m}$ among themselves and in the random variables $\{ l_{s,b}^{(F)}\}_{1 \leq s \leq b \leq m}$ among themselves, which need proper handling.

\section{Proof of \cref{prop:MS-log-likelihood-computation}}\label{appendix:Proof of prop:MS-log-likelihood-computation}

We have, 
       \begin{eqnarray} \label{eq:app-MS-LLR}
        l_{\nu,n} &=& \log \frac{\mathbb{P}_{\nu}(U_1^n, Y_1^n,J_1^n,Z_1^n|\bm{Q}_1)}{\mathbb{P}_{\infty}(U_1^n, Y_1^n,J_1^n,Z_1^n|\bm{Q}_1)} \nonumber\\
        &=& \sum_{k=1}^{n} \log \frac{\mathbb{P}_{\nu}(Y_{k}| Y_1^{k-1}, \bm{Q}_1)}{\mathbb{P}_{\infty}(Y_k| Y_1^{k-1},\bm{Q}_1)} \nonumber\\
        \quad &+& \sum_{k=1}^{n} \log \frac{\mathbb{P}_{\nu}(U_k, J_k | U_1^{k-1}, Y_1^{k}, J_1^{k-1}, \bm{Q}_1)}{\mathbb{P}_{\infty}(U_k, J_k | U_1^{k-1}, Y_1^{k}, J_1^{k-1}, \bm{Q}_1)} \nonumber\\
        \quad &+& \sum_{k=1}^{n} \log \frac{\mathbb{P}_{\nu}(Z_k | U_1^{k}, Y_1^{k}, J_1^{k}, Z_1^{k-1},\bm{Q}_1)}{\mathbb{P}_{\infty}(Z_k | U_1^{k}, Y_1^{k}, J_1^{k}, Z_1^{k-1},\bm{Q}_1)} 
   \end{eqnarray}

The probability term in first summation can be written as
\begin{align*}
     \mathbb{P}(Y_k| Y_1^{k-1}, \bm{Q}_1) 
     = \mathbb{P}(Y_k | \|\bm{Q}_k\|_0 > 0 ) \mathbb{P}(\|\bm{Q}_k\|_0 > 0 \mid  Y_1^{k}, \bm{Q}_1) 
\end{align*}
The second equality is valid because $Y_k$ depends only on $\bm{Q}_k$.
The second product term in the last equality is same under both probability laws.
The term in first summation of \ref{eq:app-MS-LLR} then simplifies to,
\begin{align*}
    \log \frac{\mathbb{P}_{\nu}(Y_k| Y_1^{k-1}, \bm{Q}_1)}{\mathbb{P}_{\infty}(Y_k| Y_1^{k-1},\bm{Q}_1)} = \mathbbm{1}_{\{\|\bm{Q}_k\|_0 > 0\}} \log \frac{\mathbb{P}_{\nu}(Y_k| \|\bm{Q}_k\|_0 > 0)}{\mathbb{P}_{\infty}(Y_k| \|\bm{Q}_k\|_0 > 0)}
\end{align*}
Since $P(Y_k= \phi)=1$ on the set $\{\bm{Q}_k = \bm{0}\}$, hence the indicator function.

The sampling instant of sensor measurement $J_k$ is independent of the change point. The sensor selection process depends only on the queue states and the contention mechanism, and hence is also independent of the change point. Hence, $\mathbb{P}(J_{k}, Z_{k} | U_1^{k}, Y_1^{k}, J_1^{k-1},\bm{Q}_1)$ is identical under both $\mathbb{P}_{\infty}$ and $\mathbb{P}_{\nu}$, and the second term in \eqref{eq:app-MS-LLR} vanishes. The last term in
\eqref{eq:app-MS-LLR} can be written as, \small
\begin{align*}
    & \log \frac{\mathbb{P}_{\nu}(Z_k | U_1^{k}, Y_1^{k}, J_1^{k}, Z_1^{k-1},\bm{Q}_1)}{\mathbb{P}_{\infty}(Z_k | U_1^{k}, Y_1^{k}, J_1^{k}, Z_1^{k-1},\bm{Q}_1)} = \log \frac{\mathbb{P}_{\nu}(Z_k | U_{k}, Y_{k}, J_{k},\bm{Q}_1)}{\mathbb{P}_{\infty}(Z_k | U_{k}, Y_{k}, J_{k},\bm{Q}_1)} \\
    & \quad \quad \quad \quad \quad \quad \quad \quad = \mathbbm{1}_{\{Y_k=1, J_k > Q_{1, U_k} + S_{\nu, U_k}\}} \log \frac{f_1(X_{J_k,U_k})}{f_0(X_{J_k,U_k})} 
\end{align*}
\normalsize
If we already know $U_k,J_k$, we have identified the specific observation $Z_k$, therefore it is independent of past sensor indices and sampling indices.
Finally, we have,
\small
\begin{align*}
    l_{\nu,n} &= \sum_{k=1}^{n} \bigg( \mathbbm{1}_{\{\|\bm{Q}_k\|_0 > 0\}} \log \frac{\mathbb{P}_{\nu}(Y_k| \|\bm{Q}_k\|_0 > 0)}{\mathbb{P}_{\infty}(Y_k| \|\bm{Q}_k\|_0 > 0)} \\
    &\qquad \qquad \qquad \qquad +  \mathbbm{1}_{\{Y_k =1,  J_k > Q_{1, U_k} + S_{\nu, U_k}\}}\log \frac{f_1(X_{J_k,U_k})}{f_0(X_{J_k,U_k})} \bigg)
\end{align*}
\normalsize
The indicator in the second term ensures that only those observations are taken into account which occur after the change instant, $\nu$.

\section{Proof of \cref{prop:MS-I-non-homogeneous}}\label{appendix:Proof of prop:MS-I-non-homogeneous}

We have,
\footnotesize
\begin{align}\label{eq:MS-bar-I}
    \bar{I} &\coloneqq \displaystyle \lim_{n \to \infty} \frac{1}{n} l_{0,n} \nonumber\\
    &=\lim_{n \to \infty} \frac{1}{n}  \sum_{k=1}^n \mathbbm{1}_{\{Y_k=1 \}} \log \frac{f_{1,U_k}(X_{J_k,U_k})}{f_{0,U_k}(X_{J_k,U_k})} \nonumber \\
    & = \mathbb{E}_{\Pi_{\zeta}^{(p)}} \left[\mathbbm{1}_{\{Y=1\}} \log \frac{f_{1,U}(X_{J,U})}{f_{0,U}(X_{J,U})} \right] \nonumber\\
    &= \mathbb{P}(Y=1) \mathbb{E}_{\Pi_{\zeta}^{(p)}}\left[\log \frac{f_{1,U}(X_{J,U})}{f_{0,U}(X_{J,U})} \mid (Y=1) \right] \nonumber\\
    &= \mathbb{P}(Y=1) \times \nonumber\\
    &\quad\sum_{i=1}^{L} \mathbb{P}(U = i \mid Y=1) \mathbb{E}_{\Pi_{\zeta}^{(p)}} \left[\log \frac{f_{1,i}(X_{J,i})}{f_{0,i}(X_{J,i})} \mid (U=i, Y=1) \right] \nonumber\\
    & = \mathbb{P}(Y=1)  \sum_{i=1}^{L} \frac{r_i}{\sum_j r_j} I(f_{1,i},f_{0,i}),
\end{align}
\normalsize
where in steady state, the fraction of successful transmissions originating from sensor $i$ is proportional to its sampling rate, i.e., $\mathbb{P}(U = i \mid Y = 1) = \frac{r_i}{\sum_j r_j}$. This follows because, over the long run, each sensor $i$ contributes samples at rate $r_i$, and the DM receives $\sum_j r_j$ samples per unit time in total. 

Now, $\mathbb{P}(Y=1)$ represents the steady-state probability that a transmission is attempted and succeeds. To model this, consider a “super queue” whose arrival rate is $\sum_{i=1}^{L} r_i$ (the aggregate sampling rate across all sensors), and with Bernoulli$(p)$ departures. The stability condition yields that the probability the super queue is non-empty is $\frac{\sum_{i=1}^L r_i}{p}$, so the long-run probability of successful transmission is
\[
\mathbb{P}(Y=1) = p \cdot \frac{\sum_{i=1}^L r_i}{p} = \sum_{i=1}^L r_i.
\]
Substituting in \eqref{eq:MS-bar-I}, we obtain
\[
\bar{I} = \sum_{i=1}^{L} r_i I(f_{1,i},f_{0,i}),
\]
which proves the proposition.

\end{document}